\documentclass[12pt]{article}
\usepackage[utf8]{inputenc}
\usepackage{amsmath,amssymb,amsthm,amsfonts,amscd, bm}
\usepackage[margin = 1in]{geometry}
\usepackage{setspace}
\usepackage{xcolor}
\usepackage{amsmath,amssymb,amsthm,bbold,bm}
\usepackage{booktabs}
\usepackage{adjustbox}
\usepackage{tikz-network}
\PassOptionsToPackage{hyphens}{url}
\usepackage[colorlinks,breaklinks,   
     bookmarks=false,
     pdfstartview=Fit,  
     pdfview=Fit,       
     linkcolor=aspurple,
     urlcolor=aspurple,
     citecolor=aspurple,
     anchorcolor = aspurple,
     hyperfootnotes=false
         ]{hyperref}
\usepackage[font = normalsize, labelsep = period, skip=5pt, labelfont = bf]{caption}
\usepackage{cleveref}
\usepackage{authblk}
\usepackage{natbib}

\input{glyphtounicode}\pdfgentounicode=1
\definecolor{aspurple}{RGB}{51,9,128} 
\newtheorem{theorem}{Theorem}
\newtheorem{proposition}[theorem]{Proposition}
\theoremstyle{definition}

\newcommand{\figtext}[1]{
	\captionsetup{justification=justified,font=footnotesize, margin=0pt}
	\caption*{\hspace{0pt}\hangindent=0em #1}
	}

\newcommand{\Fignote}[1]{\figtext{\textit{Note:~}~#1}}
\def\sym#1{\ifmmode^{#1}\else\(^{#1}\)\fi}
\newtheorem*{assumption*}{\assumptionnumber}
\providecommand{\assumptionnumber}{}


\numberwithin{equation}{section}
\begin{document}

\renewcommand{\thefootnote}{\fnsymbol{footnote}}
\begin{titlepage}
    \centering
    \vspace*{1cm}
    \textbf{\huge Inflation in Disaggregated Small Open Economies\footnote{I am deeply indebted to \c{S}ebnem Kalemli-\"{O}zcan,  Pierre de Leo and Thomas Drechsel for their invaluable encouragement, guidance, and support. I am also profoundly grateful to Luna Bratti, Julian di Giovanni, Jorge Miranda-Pinto, Felipe Saffie, and Muhammed A. Y\i ld\i r\i m for their support, patience, and feedback. I would like to thank Boragan Aruoba, Sina Ates, Colin Hottman, Guido Kuersteiner, Robbie Minton, John Shea, Luminita Stevens, as well as various seminar participants at the Bank of Canada, Universidad de Chile, Pontificia Universidad Cat\'olica de Chile, Federal Reserve Board, Federal Reserve Bank of Boston, and the University of Maryland, for helpful comments and discussions.   All errors are my own. The views expressed herein are those of the authors and not necessarily those of the Federal Reserve Bank of Boston or any other person affiliated with the Federal Reserve System. Email: \href{asilvub@gmail.com}{asilvub@gmail.com}. First version: October 27, 2023.}}
    
    \vspace{0.5cm}
    \large \textbf{Alvaro Silva} \\
    \small Federal Reserve Bank of Boston\\
    \small \href{asilvub@gmail.com}{asilvub@gmail.com}
    
    \vspace{0.5cm}
    
    \normalsize \today \\
    \small \href{https://asilvub.github.io/assets/papers/AS_JMP.pdf}{\textbf{Click here for the latest version}}
    \vspace{0.5cm}
    \begin{singlespace}
    \begin{abstract}
        \noindent This paper studies inflation in small open economies with production networks. I show that the production network alters the elasticity of the consumer price index (CPI) to changes in sectoral technology, factor prices, and import prices. Sectors can import and export directly but also indirectly through domestic intermediate inputs. Indirect exporting dampens the inflationary pressure from domestic forces, while indirect importing increases the inflation sensitivity to import price changes. Computing these CPI elasticities requires knowledge of the production network structure as these do not coincide with typical sufficient statistics used in the literature, such as sectoral sales-to-GDP ratios, factor shares, or imported consumption shares. Using input-output tables, I provide empirical evidence that adjusting CPI elasticities for indirect exports and imports matters quantitatively for small open economies. I use the model to illustrate the importance of production networks during the recent COVID-19 inflation in Chile and the United Kingdom.\\[0.25in]



        \end{abstract}
    \end{singlespace}
\end{titlepage}
\renewcommand{\thefootnote}{\arabic{footnote}}

\newpage
\clearpage
\pagenumbering{arabic} 
\section{Introduction}

In 2022, inflation reached 8 percent in the United States, its highest level in 40 years. The picture was similar on the other side of the Atlantic: Euro Area inflation was 8.4 percent, the highest since its creation.  Explanations include shocks to commodity prices \citep{BB23a, GG23}, sectoral demand changes \citep{FGI22}, fiscal stimulus \citep{BFM23, dGKOSY23}, and supply chain disruptions \citep{dGKOSY22, dGKOSY23b, CJJ23}. As shown in Figure \ref{fig:inflation}, high inflation was not restricted to these two economies: the median small open economy experienced an inflation rate of around 10 percent in 2022. However, inflation in this group of countries has been less studied during the current episode. This paper attempts to fill this gap using both theory and data. 

\begin{figure}[htbp!]
    \centering
        \caption{CPI Inflation in Small Open Economies.}
    \label{fig:inflation}
    \includegraphics[scale = 0.6]{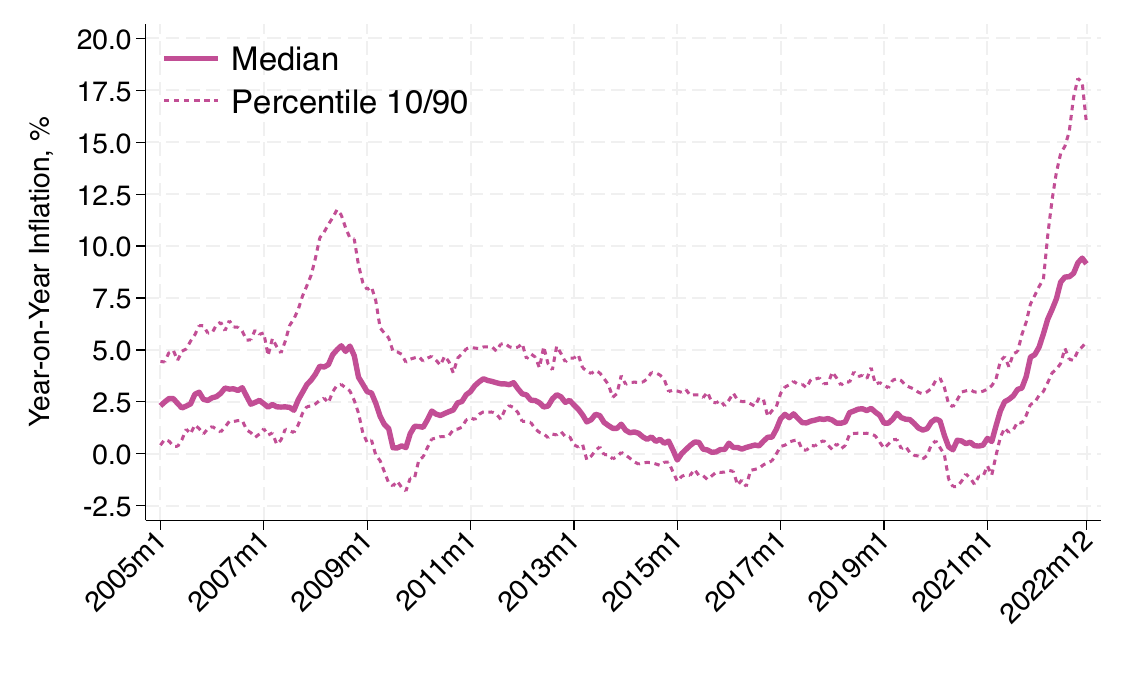}
    \Fignote{The figure shows the median inflation rate (solid lines) and 10th and 90th percentile (dashed lines). Small open economies are economies that represent less than 5 percent of world GDP and have a trade openness larger than 30 percent of GDP. See section \ref{subsec:data} for more details.  Plot shows an unbalanced panel of 46 small open economies over time. Source: Bank for International Settlements.}
\end{figure}




My starting point is the multi-sector and multi-factor production network closed economy model in \cite{BF19}. It provides a useful benchmark to analyze inflation during macroeconomic shocks such as COVID-19, a combination of sectoral and aggregate shocks. Given my focus on small open economies,  I augment this model to feature imports and exports at the sectoral level, adapting the production network model to the small open economy case. I use the model to study how the consumer price index (CPI) reacts to changes in sectoral technology, factor prices, and import prices, going from the micro to the macro level. 

I show that openness and production networks affect our understanding of inflation in small open economies via two distinct channels. On the one hand, exporting, either directly or indirectly through other economic producers, dampens the effect of sectoral technology shocks and factor price changes relative to a closed economy. On the other hand, direct importing gives rise to the problem of imported inflation as the domestic consumer's basket now contains imported goods. On top of this channel, production networks imply that domestic goods are manufactured using imported inputs indirectly. As a result, production networks amplify imported inflation.\footnote{This channel is distinct from inflation resulting from imported intermediate goods, as models can have intermediate goods without intersectoral linkages. See \cite{Svensson00} for an early analysis of imported inflation via intermediate goods.} Uncovering these effects and quantifying their importance is only possible when both openness and network linkages are explicitly considered.

The key economic intuition is that opening up the economy is one of the ways to break the link between what the country produces and what is consumed by domestic consumers. In an efficient closed economy --- an economy without distortions --- with intersectoral linkages and domestic final consumption only, everything produced is consumed by the domestic consumer. \emph{Network-adjusted domestic consumption}, by which I mean domestic consumption adjusted by domestic production network linkages, is thus equivalent to sales in the closed economy.\footnote{This definition is deliberately reminiscent of the network-adjusted labor share introduced in \cite{Baqaee15}.} That domestic households consume everything produced, directly or indirectly, is one of the key building blocks of why the production network structure is irrelevant to first-order for macroeconomic outcomes such as real GDP or welfare in closed economies. This irrelevance result is a useful benchmark. It allows us to use the ratios of sectoral sales to nominal GDP  (the so-called ``Domar weights") and factor payments to nominal GDP (factor shares)  as sufficient statistics for the pass-through of sectoral technology changes or factor price changes to the CPI, respectively.\footnote{As I show in the theory section, this can be thought of as a corollary of Hulten's theorem \citep{Hulten78} but for the CPI rather than for real GDP. Recall that for real GDP, Hulten's theorem states that in an efficient closed economy with inelastic factor supplies, the first-order effect of sectoral technology on real GDP is given by the Domar weights, and the first-order effect of changes in factor supply is given by its factor share. } Increases in sectoral technology decrease consumer inflation by the Domar weight of the sector, while increases in factor prices increase inflation by the factor share.


I show that this irrelevance result no longer holds for consumer inflation in small open economies, without the need for second-order approximations, as in \cite{BF19}, or distortions, as in \cite{BF20} and \cite{BL20}.\footnote{Strictly speaking, those papers I cited sought ways to break Hulten's theorem in closed economies, which refer to quantities, meaning the effect of sectoral technology changes or distortions on real GDP. However, since, as a corollary of Hulten's theorem, we can back out changes in CPI, I referenced them here.} The reason is as follows. Consider first the impact of sectoral technology shocks on the CPI. In a small open economy, there are two final uses for goods produced within borders: domestic consumption or exports. Unlike the closed economy case, sectoral sales do not map to the network-adjusted domestic consumption for two reasons: (i) direct exports and (ii) indirect exports through domestic production network linkages. Instead, sectoral sales map to network-adjusted domestic consumption \emph{plus} network-adjusted exports. Since what matters for the CPI is network-adjusted domestic consumption, one must subtract network-adjusted exports from sectoral sales. Hence, relative to a closed economy, a consumer is weakly less exposed to changes in sectoral technology. Importantly, we require knowledge of the domestic production network structure to compute these network-adjusted domestic consumption measures. 


A similar intuition holds for how factor price changes affect CPI. The relevant statistic here is the \emph{network-adjusted domestic factor share}: how much of each factor is embedded in goods consumed domestically after considering domestic production network linkages (in the spirit of the domestic factor demand concept in \cite{ACCDP22}). The total amount of a factor available in the economy can be ``consumed" by domestic or foreign consumers,\footnote{Here, consumers do not directly consume factors but goods. Given that goods are ultimately made of factors of production, we can think of consumers implicitly consuming them. This notion can be found in the reduced factor demand system proposed by \cite{ACD17}.} with the production network potentially reshaping these patterns. While factor shares are sufficient statistics in the closed economy, in the small open economy with production networks we need to subtract from factor shares the fraction of each factor that is exported either directly or indirectly via production networks. This means that relative to a closed economy, the domestic consumer is weakly less exposed to changes in factor prices. 

The effect of import prices on the CPI, on the other hand, is amplified in a small open economy with production networks relative to a small open economy without production networks. The relevant statistics here are \emph{network-adjusted import consumption shares}. Since the domestic consumer directly imports, the direct consumption share captures part of the exposure to import price changes. However, if there are intersectoral linkages across producers, domestic good producers may end up importing intermediate inputs either directly, by buying from abroad, or indirectly, by buying from domestic sectors that buy from abroad or that buy from sectors that buy from abroad, and so on. This means that the imported content of domestically produced goods increases in the presence of production networks. To the extent that domestic goods increase their reliance on imported intermediate goods, so does the domestic consumer. Thus, the domestic consumer's exposure to import prices must account for both direct and indirect exposure, which are encapsulated in the \emph{network-adjusted import consumption shares}.




Guided by the model, I turn to the data to measure the importance of these production network adjustments. I find that these adjustments matter quantitatively using data from the World Input-Output Tables. I illustrate these adjustments by focusing on the three sources of variation I have considered so far: sectoral technology, factor prices, and import prices. 

First, consider the electricity sector in the United Kingdom (UK). The Domar weight of this sector is around 5.95 percent. Once we consider direct exports (but not indirect exports), the relevant ratio for the pass-through to CPI decreases to 5.90 percent, a negligible change. This is expected, as the UK electricity sector hardly exports directly to other countries. Yet, when considering indirect exporting, the network-adjusted domestic consumption share decreases to 4.4 percent, a 25 percent decrease relative to the Domar weight benchmark. This is because other export-heavy sectors in the UK use electricity as a production input either directly or indirectly. Thus, Domar weights would overestimate the impact of a change in productivity in the UK electricity sector on the domestic CPI.

Second, consider the role of wage changes in the CPI. In a closed economy, the labor share is the relevant statistic for how wage changes pass through to the CPI. In the data, the labor share for the average small open economy is around 57 percent. However, the small open economy model with production networks suggests that we need to subtract from the labor share the portion that is exported directly or indirectly. After accounting for network-adjusted exports, this average labor share decreases to 39 percent. This means that the same increase in domestic wages has a 32 percent lower impact on inflation in a small open economy relative to a closed economy. 

Finally, let me consider the role of import prices.  In the data, the average small open economy exhibits a direct import consumption share of around 17 percent of its total expenditure. Yet, on average, the network-adjusted import consumption share is 30 percent. This implies that the impact of import prices on domestic inflation is (almost) twice what would be implied by a measure ignoring indirect linkages.

In the last section of the paper, I use the model to analyze the recent inflation in two small open economies: Chile and the United Kingdom. I chose these two countries as (i) they fit into the small open economy definition, (ii) they have experienced high inflation in recent years, and (iii) they allow me to compute and contrast between emerging and developed markets. Using these countries,  I show that network adjustment on exports and imports provides a quantitative improvement in matching data moments over both closed economy models and small open economy models without production networks. 

Between 2020 and 2022, the average annual inflation rate in Chile was 6.13 percent, with a standard deviation of 3.89. A quantitative closed economy model with production networks implies an average inflation rate of 0.98 percent with a standard deviation 9.69. A small open economy model without production networks delivers an average inflation rate of 1.45 percent with a standard deviation 6.88. Finally, the small open economy model with production networks delivers an average inflation rate of 2.41 with a standard deviation 6.67.  Overall, the small open economy with production networks better matches the mean and the standard deviation.

For the United Kingdom, average inflation rate was 3.69 percent over the same period, with a standard deviation of 3.11. The closed economy model with production networks implies an average inflation rate of 2.27 percent with a standard deviation 2.57. The small open economy model without network adjustments exhibits an average inflation rate of 2.72 percent with a standard deviation of 2.64. The small open economy model with production networks shows an average inflation rate of 3.21 percent with a standard deviation 3.00. As in the Chilean case, the production network coupled with openness helps to get closer to the date moments of United Kingdom inflation.

The measurement and application sections illustrate that the required network adjustments in a small open economy matter for our understanding of inflation not only as a theoretical curiosity but also in practice.

\paragraph{Related Literature.} This paper relates to several strands of the literature. The first one studies inflation in closed economies with production networks \citep{Basu95, LTS22, GLSW22, BF22,  LV23, AB22, FGI22, Rubbo23, MW23, dGKOSY23b, dGKOSY23, LW23}.\footnote{There is also extensive literature on multisector models with sticky prices that do not necessarily feature a production network structure, so I omit them from the main text. For earlier contributions, see \cite{Woodford03} and the references therein. } These studies consistently find that the interaction between sectoral price/wage rigidities and production networks is key to understanding the behavior of inflation, which has implications for the conduct of monetary policy, such as what inflation rate to target. This paper focuses on how introducing production networks in a small open economy model helps us to understand the pass-through of different shocks to inflation. Although there is no price rigidity in the model, and thus I cannot speak about the optimal conduct of monetary policy, I contribute to this literature by showing that the production network can have  a first-order impact on inflation beyond its role in the sales share distribution, without the need for any distortions. 

Second, this paper relates to the literature on inflation in small open economies. In the second part of the 20th century, Latin America experienced episodes of high and persistent inflation, a term later coined as ``chronic inflation". In response, there was extensive literature during the 1990s on how to best control chronic inflation and the impact of different nominal and real policy rules in small open economies \cite[see][and especially the last one, for an overview of this earlier literature]{CV95, CRV95, CV99}. Modern treatments that introduce New Keynesian features such as sticky prices and monopolistic competition into small open economy models include \cite{GM05} and \cite{FM08}\footnote{There is also a large literature focusing on two or more countries. My work is not directly related to these models as I focus on small open economies. I refer the interested reader to \cite{CP07} and \cite{CDL10} for an overview of such models. Recent literature focusing on inflation using multi-country and multi-sector models include, for example, \cite{ALS19} and \cite{HSS22} during non-Covid-19 times and \cite{dGKOSY22} and \cite{AAS23} during COVID-19.}. The literature has recently augmented these models to include multiple sectors \citep{Matsumara22} and intersectoral linkages \citep{QWXZ24} and has applied these models to understand the recent inflationary episode in the United States \citep{CJ20, CJJ23}. Relative to this literature, I make four contributions. First, I explicitly analyze the role of production networks on inflation for small open economies. I show how the production network interacts with international trade, affecting how domestic and foreign shocks ultimately affect CPI inflation both theoretically and quantitatively. Second, I show that this result holds without the need of any distortion or friction, which further clarifies the key role of production networks beyond these forces. Third, the fact that I use a first-order approximation allows me to consider unrestricted intersectoral linkages, without assuming any functional forms for production or utility. Finally, my model also features multiple factors of production, while the previous models typically focus on only one factor, labor.


Finally, in focusing on the role of network-adjusted exports and imports, this paper contributes and connect inflation to the literature on indirect trade \citep{Huneeus18, ACCDP22, DKMT21,  DKMT23, Munoz23}. This literature focuses on the firm-level real consequences of indirect trade, which is equivalent to my trade network-adjustments.  For example, \cite{DKMT21} use Belgian firm-to-firm level transaction data and find that the relevant concept for a firm sales' exposure to international markets is total exports (network-adjusted exports), while its exposure in costs is total imports (network-adjusted import share). Importantly, the focus of the trade literature is on real outcomes, while the focus of this paper is on a nominal variable. Hence, my contribution to this literature is to embed indirect trade into a small open economy model to analyze how it matters for a nominal variable, inflation.



\section{A Small Open Economy Model with Production Networks} \label{sec:static_model}



\textbf{Environment.} There is a set of domestically produced goods that I denote by $N$ with typical element $i$. These goods can be consumed domestically, used as intermediate inputs by other domestic sectors, and exported. I denote the imported goods set by $M$, with typical element $m$. These imported goods can be used as intermediate inputs to produce domestic goods or as final consumption. Finally, there is a set $F$ of factors with typical element $f$.

\paragraph{Notation.} I denote matrices and vectors using \textbf{bold} i.e., $\bm{Y}$. I denote the transpose of a matrix as $\bm{Y}^T$. Unless otherwise noted, vectors are always column vectors. For example, the vector of Domar weights, defined below, is $\bm{\lambda} = (\lambda_1, \lambda_2,..., \lambda_N)^T$. Log changes are expressed as $\mathrm{d}\log Y = \hat{Y}$.

Table \ref{tab:definitions} shows the different shares and matrices that are key for the analysis. I use a bar over a variable for shares based on \emph{total expenditure}, while GDP-based measures do not contain a bar.

\begin{table}[h!]
    \caption{Definitions}
    \renewcommand{\arraystretch}{1.7}
\label{tab:definitions}
    \centering
\begin{adjustbox}{max width = 0.9\textwidth}
\begin{tabular}{l c c c}
\toprule 
Name & Notation &  Expression & Goods/Factors\\
\midrule 
\emph{GDP-based}\\
Domar Weight & $\lambda_i$ & $\frac{P_iQ_i}{GDP}$ & for  $i  \in N$ \\
Consumption Share & $b_i$ & $\frac{P_iC_i}{GDP}$ & for  $i  \in N$\\
Imported Consumption Share & $b_m$ & $\frac{P_mC_m}{GDP}$ & for  $m\in M$\\
Export Share & $x_i$ & $\frac{P_iX_i}{GDP}$ & for  $i \in N$ \\
Factor Shares & $\bm{\Lambda}$ & $\Lambda_f = \frac{W_fL_f}{GDP}$ & for  $f\in F$\\
\\
\emph{Expenditure-based}\\
Domar Weight & $\bar{\lambda}_i$ & $\frac{P_iQ_i}{E}$ & for  $i  \in N$ \\
Consumption Share & $\bar{b}_i$ & $\frac{P_iC_i}{E}$ & for  $i  \in N$\\
Imported Consumption Share & $\bar{b}_m$ & $\frac{P_mC_m}{E}$ & for  $m \in M$\\
Export Share & $\bar{x}_i$ & $\frac{P_iX_i}{E}$ & for  $i \in N$ \\
Factor Shares & $\bar{\bm{\Lambda}}$ & $\bar{\Lambda}_f = \frac{W_fL_f}{E}$ & for  $f\in F$\\
\\
\emph{Sector-level Shares}\\
Input-Output Matrix & $\bm{\Omega}$ & $\Omega_{ij}= \frac{P_jM_{ij}}{P_iQ_i}$ &  $j \in N$\\
Leontief-Inverse Matrix & $\bm{\Psi}_D = (\bm{I} - \bm{\Omega})^{-1}$ & $\bm{\Psi}_{ij} = \sum\limits_{s=0}^\infty \bm{\Omega}^s_{ij}$&    $i,j \in N$\\
Factor Spending Matrix & $\bm{A}$ & $a_{if} = \frac{W_fL_{if}}{P^D_i Q_i}$ & $i\in N; f \in F$\\
Intermediate Import Spending Matrix & $\bm{\Gamma}$ & $\Gamma_{im} = \frac{P_mM_{im}}{P_iQ_i}$ & $i \in N; m \in M$ \\
\bottomrule 
\end{tabular}
\end{adjustbox}

\end{table}

\subsection{Households} 
There is a representative household that consumes domestically produced goods and foreign goods. It has an instantaneous utility function that I denote by $U(\bm{C}_D, \bm{C}_M)$, where $\bm{C}_D = \lbrace C_i\rbrace_{i \in N}$ denotes the vector of domestically-produced goods consumption and $\bm{C}_M = \lbrace C_m\rbrace_{m \in M}$ is the vector of foreign goods consumption. These consumption vectors have associated vector prices  $\bm{P}_D = \lbrace P_i\rbrace_{i \in N}$ and $\bm{P}_M = \lbrace P_m\rbrace_{m\in M}$. Unless otherwise stated, all prices are denominated in local currency. I assume the utility function $U(.)$ is homogeneous of degree one in its arguments. The representative consumer also owns all factors of production and supplies them inelastically ($\lbrace \bar{L}_f\rbrace_{f\in F}$) at the given factor prices ($\lbrace W_f \rbrace_{f\in F}$).

Given a vector of prices, both domestically produced and foreign goods, the cost-minimization problem satisfies
\begin{align}
   PC =  \min_{\bm{C}_D, \bm{C}_M}  \sum\limits_{i\in N}P_iC_i + \sum\limits_{m\in M} P_mC_m\text{ subject to } U(\bm{C}_D, \bm{C}_M)\geq \bar{U}.
\end{align}

Solving this problem delivers a price index that is a function of good prices. I denote this price index by $P = P(\bm{P}_D, \bm{P}_M)$. As a reminder, notice that up to a first-order approximation, changes in this price index satisfy 
\begin{align}
    \widehat{P} &= \bar{\bm{b}}_D^T \widehat{\bm{P}}_D + \bar{\bm{b}}_M^T \widehat{\bm{P}}_M,
\end{align}
where 
\begin{align*}
    \bar{\bm{b}}_D = \lbrace \bar{b}_i\rbrace = \frac{P_iC_i}{E}; \quad  \bar{\bm{b}}_M = \lbrace \bar{b}_m\rbrace = \frac{P_mC_m}{E}; \quad E = \bm{P}_D^T \bm{C}_D + \bm{P}_M^T \bm{C}_M = PC,
\end{align*}
are the expenditure share on domestically produced goods ($\bar{b}_i$), imported goods ($\bar{b}_m$), and total expenditure ($E$), respectively.

The consumer budget constraint reads 
\begin{align*}
      PC + T  &= \sum\limits_{f\in F} W_f L_{f} + \sum\limits_{i\in N} \Pi_i,  
\end{align*}
where $T$ is an \emph{exogenous} net transfer to the rest of the world as in \cite{BF22}. In Appendix \ref{app:two_period}, I provide a justification for having such a force in the current model using a two-period model without changing the main results.

\subsection{Sectors}
There is a representative firm in each $i$ sector that produces according to the following production function 
\begin{align}
    Q_i &= Z_i F^i\left(\lbrace L_{if}\rbrace_{f\in F}, \lbrace M_{ij}\rbrace_{j\in N}, \lbrace M_{im}\rbrace_{m\in M}\right),
\end{align}
where $Z_i$ is a sector-specific productivity, $L_{if}$ is demand for factor $f$ by firm $i$, $M_{ij}$ represents intermediate input demand for good $j \in N$ by firm $i$, and $M_{im}$ represents input demand for imported good $m\in M$. We can write cost-minimization firm $i$ as 
\begin{align*}
    TC_i = \min_{\lbrace L_{if}\rbrace_{f=1}^F, \lbrace M_{ij}\rbrace_{j\in N}, \lbrace M_{im}\rbrace_{m\in M}} &\sum\limits_{f\in F} W_f L_{if} + \sum\limits_{j\in N} P_j M_{ij} + \sum\limits_{m\in M} P^M_m M^M_{im}\\
    & \text{ subject to } Z_i F^i\left(\lbrace L_{if}\rbrace_{f\in F}, \lbrace M_{ij}\rbrace_{j\in N}, \lbrace M_{im}\rbrace_{m\in M}\right)\geq \bar{Q}_i.
\end{align*}

This delivers a marginal cost function that only depends on prices and technology due to the constant returns to scale assumption. In particular, 
\begin{align}
    MC_i = MC_i(Z_i, \bm{P}_D, \bm{P}_M, \bm{W}),
\end{align}
where $\bm{W} = \lbrace W_f\rbrace_{f\in F}$ is a vector of factor prices. 

We can get conditional factor and intermediate input demand by applying Shephard's lemma to the optimized total cost, $TC_i$, such that 
\begin{align}
    \frac{\partial MC_i}{\partial W_f}Q_i&=L_{if}\quad \text{ for each } f \in F,\\
    \frac{\partial MC_i}{\partial P_j}Q_i&=M_{ij} \text{ for each } j \in N,\\
    \frac{\partial MC_i}{\partial P_m}Q_i&=M_{im} \text{ for each } m \in M.
\end{align}
Due to constant returns to scale and perfectly competitive good and factor markets, each firm $i$ makes zero profit:
\begin{align}
    P_iQ_i&= \sum\limits_{f\in F} W_f L_{if} + \sum\limits_{j\in N} P_j M_{ij} + \sum\limits_{m\in M} P_m M_{im}\quad \text{ for all } i \in N.
\end{align}

\subsection{Equilibrium}
Market clearing conditions for good and factor markets satisfy 
\begin{align}
    Q_i &= C_i + X_i + \sum\limits_{j\in N} M_{ji} \qquad \text{ for each } i \in N \label{eq:GMC}.
\end{align}
Equation (\ref{eq:GMC}) is the good market clearing condition. I assume $X_i$ is exogenous as in \cite{ACCDP22} so that a price clearing the market always exists for each domestically produced good even if it is exported. 

Since this is a real model, nominal prices are indeterminate unless I supplement one additional equation. To do so, I impose the following
\begin{align*}
    PC &\leq \mathcal{M} = E,
\end{align*}
where $\mathcal{M}$ is the money supply that I take as exogenous in what follows. This is a cash-in-advance constraint used, for example, in \cite{LTS22} and \cite{AB22}.\footnote{It can be shown that this ``constraint" is isomorphic to a model with money in the utility function that is separable from aggregate consumption. The cash-in-advance constraint thus serves no other purpose than pinning down nominal variables without affecting real allocations in this model.} We can think of this restriction as the small open economy's central bank effectively pinning down total nominal expenditure ($E$), providing an exogenous nominal anchor. It is apparent that the central bank, conditional on knowing $C$, which is determined by \emph{real variables}, can implement any price level, $P$, that it desires consistent with $C$. This model features the classical dichotomy, where real variables are determined independently of the nominal side.\footnote{The converse is not true as real shocks can affect nominal variables. See \cite{Vegh2013} chapter 5, especially footnote 11.} Under these assumptions, one should interpret the results as highlighting the role of production networks for the consumer price index, conditional on an exogenous central bank monetary policy.


Similar to \cite{BF19BBVA}, I define an equilibrium in this economy using a dual approach in which feasible and equilibrium allocations are found by taking as given factor prices $\bm{W}$ and a level of expenditure, $E$, as follows 
\begin{enumerate}
    \item Given sequences ($\bm{W},\bm{P}_D,\bm{P}_M, \bm{\Pi}$)  and exogenous parameters ($T$), the household chooses $(\bm{C}_D,\bm{C}_M)$ to maximize its utility subject to its budget constraint.
    \item Given ($\bm{W},\bm{P}_D,\bm{P}_M$) and production technologies, firms choose $(\bm{L}_i, \bm{M}_i)$ to minimize their cost of production.
    \item Given $\bm{X}$, goods markets clears.
    \item The cash-in-advance constraint holds with equality $PC = \mathcal{M} = E$
\end{enumerate}


\subsection{Characterizing Changes in the Price Index}

Having defined the environment, optimality, and equilibrium conditions, I can now study changes in the consumer price index, $\widehat{P}$.  Inflation here consists of a log-linear approximation around the initial price level equilibrium. The purpose of the model is to distill whether and how the production network may matter for inflation, which in the model is a cross-sectional statement rather than a dynamic statement. ``Inflation" in this context can thus be understood in the space rather than the time dimension. This concept has been used, for example, to study inflation in the US during the COVID-19 period \citep{BF22, dGKOSY22}, and the role of sticky prices in production networks \citep{LTS22, BR22}.

The following result characterizes how the consumer price index reacts to changes in exogenous variables.

\begin{proposition} \label{prop1}
Consider a perturbation $(\widehat{\bm{Z}},\widehat{\bm{W}}, \widehat{\bm{P}}_M)$ around some initial equilibrium. Up to a first order, changes in the aggregate price index, $\widehat{P}$, satisfy
\begin{align}
     \widehat{P} &= -\left(\bm{\bar{\lambda}}^T - \bm{\widetilde{\lambda}}^T \right)\widehat{\bm{Z}} + \left(\bm{\bar{\Lambda}}^T - \bm{\widetilde{\Lambda}}^T \right)\widehat{\bm{W}} + (\bar{\bm{b}}_M^T + \widetilde{\bm{b}}_M^T)\widehat{\bm{P}}_M\label{eq:main},
\end{align}
where 
\begin{align*}
    \bm{\widetilde{\lambda}}^T&= \bar{\bm{x}}^T\bm{\Psi}_D;\qquad 
    \bm{\widetilde{\Lambda}}^T= \bar{\bm{x}}^T\bm{\Psi}_D\bm{A};\qquad
    \widetilde{\bm{b}}_M^T= \bm{\bar{b}}_D^T\bm{\Psi}_D\bm{\Gamma}
\end{align*}
\end{proposition}
\begin{proof}
See Appendix \ref{proof:prop1}.
\end{proof}

The above expression highlights how opening up the economy to goods trade and introducing a production network structure alter the usual prediction of closed economy models. I now proceed with some illustrations that provide intuition for this expression.

\paragraph{Illustration 1: Closed economy.} The following proposition characterizes CPI in a closed economy. 

\begin{proposition}\label{prop2}
In a closed economy, equation (\ref{eq:main}) reduces to 
\begin{align*}
   \widehat{P} &=  -\bm{\lambda}^T\widehat{\bm{Z}} + \bm{\Lambda}^T\widehat{\bm{W}},
\end{align*}
\end{proposition}
\begin{proof}
See Appendix \ref{proof:prop2}.
\end{proof}
\Cref{prop2} shows the exact form of changes in the CPI in a closed economy \citep[see][]{BF22}. Intuitively, CPI changes are a weighted average of changes in productivity (weighted by the Domar weights, $\bm{\lambda}$) and factor prices (weighted by the factor shares, $\bm{\Lambda}$). 

Equation (\ref{eq:main}) extends this for small open economies with production networks.
There are four differences between the closed economy expression and equation (\ref{eq:main}). First, of course, inflation now depends on import price changes. Second, the Domar weights and factor shares in equation (\ref{eq:main}) are now based on expenditure rather than on nominal GDP. This distinction arises in small open economies that feature trade imbalances, in which the income from domestic production need not equal what they consume. Since what matters for the CPI is what domestic consumers spend, nominal expenditure is the relevant object for dividing sales and factor payments.

Third, the effect of sectoral productivity changes on the CPI is dampened relative to a closed economy or a small open economy without production networks. To see this, note that the relevant statistic for the effect of sectoral productivity changes on the CPI is $\bm{\bar{\lambda}}^T - \bm{\widetilde{\lambda}}^T $, and thus the Domar weight $\bm{\bar{\lambda}}^T$ is no longer the sufficient statistic for understanding how sectoral productivity changes affect CPI. Importantly, the relevant elasticity requires adjusting the expenditure-based Domar weight $\bar{\bm{\lambda}}$ by substracting $\bm{\widetilde{\lambda}}^T = \bar{\bm{x}}^T\bm{\Psi}_D$. This adjustment comes from the fact that what matters is the domestic consumer's exposure to changes in sectoral productivity. 

To be precise, let me write the price index as a function of domestic and imported goods prices, i.e., $P = \mathcal{P}(\bm{P}_D,\bm{P}_M)$. Suppose there is a change in the productivity of sector $k$, $\widehat{Z}_k$, with no changes in factor or import prices. This shock impacts all domestic goods prices due to input-output linkages. Its propagation to the CPI is a tale of two elasticities. First, how exposed is the consumer to changes in domestic good prices $\frac{\partial \log \mathcal{P}}{\partial \log P_i}$, for all $i$. By the envelope theorem, this elasticity is simply the consumption share of the good at the initial equilibrium, $\bar{b}_i$. Second, the impact depends on how productivity passes through to each domestic goods price, $\frac{\partial \log P_i}{\partial \log Z_k}$. This last term is simply given by $-\Psi_{ik}$, which measures the sensitivity of the price of good $i$ to a change in productivity of sector $k$ after taking into account all direct and indirect linkages via the production network. Collecting all these pieces, we can write:
\begin{align*}
    \widehat{P}&= \sum\limits_{i\in N}\underbrace{\frac{\partial \log \mathcal{P}}{\partial \log P_i}}_{=\bar{b}_i}\underbrace{\frac{\partial \log P_i}{\partial \log Z_k}}_{=-\Psi_{ik}}\widehat{Z}_k = -\bar{b}_D^T\bm{\Psi}_{(:,\; k)}\widehat{Z}_k,
\end{align*}
where $\bm{\Psi}_{(:,\; k)}$ is the $k$th column of the Leontief-inverse matrix $\bm{\Psi}$. Note that again, the reason why $\bm{\bar{b}}^T_D\bm{\Psi}$ is not equivalent to the vector of sales share is precisely because this is not the relevant exposure of the domestic consumer in the presence of input-output linkages and international trade.

Third, the effect of factor prices on the CPI is also dampened relative to the closed economy benchmark or a small open economy without production networks. Although the logic is similar to that of how productivity changes pass through CPI, I analyze the factor price case in detail in the next example, as it also allows me to relate equation (\ref{eq:main}) to  a well-known concept in the trade literature: the factor content of exports.

\paragraph{Illustration 2: Domestic factor demand and the factor content of exports.} 
This example helps to illustrate how the fact that exports used domestic factors lowers the sensitivity of prices to changes in domestic factor prices.
Equation (\ref{eq:main}) highlights a tension between domestic factor demand and the factor content of exports, in the spirit of \cite{ACCDP22}. When an economy exports, some of its factors of production end up meeting foreign demand which, everything else equal, reduces \emph{domestic factor demand}. These factors meet foreign demand because they are used to produce domestic goods that are exported. As a result, factor price changes put less pressure on the price index\footnote{Though the factor content of exports is already well known in the trade literature, I am unaware of previous work linking this precise notion to inflation.}. Moreover, this channel is in place whenever an economy exports to the rest of the world, even if there is no production network. To see this, notice that factor payments to a given factor $f$ can be written as
\begin{align}
    W_f L_f &= \sum\limits_{i\in N} W_fL_{if} = \sum\limits_{i\in N} a_{if}\lambda_i. \label{eq:factor_payments}
\end{align}
Without intermediate inputs, the Domar weight of each sector, $\lambda_i$, is simply that sector's share in total final demand 
\begin{align}
    Q_i &= C_i + X_i \Longrightarrow \lambda_i = b_i + x_i. \label{eq:goods_mkt_clearing}
\end{align}
Combining equations (\ref{eq:factor_payments}) and (\ref{eq:goods_mkt_clearing}), I get 
\begin{align}
    W_f L_f - \underbrace{\sum\limits_{i\in N}a_{if} x_i}_{\text{Factor Content of Exports}} &= \underbrace{\sum\limits_{i\in N} a_{if}b_i}_{\text{Domestic Factor Demand}}.
\end{align}
This equation shows the tension: a rise in exports -- higher $x_i$ -- must be balanced out by a fall in domestic factor demand on the right-hand side, conditional on aggregate payments to factor $f$ being constant. This is one of the mechanisms through which exports cause domestic factor prices to put less pressure on domestic consumer prices.

When we allow for a production network and trade, sectors that do not export much directly (have low $x_i$) could end up exporting indirectly via other producers. The case I analyzed above, without a production network, is a particular case in which $\bm{\Omega} = \bm{0}_{N\times N}$ and thus $\bm{\Psi_D} = \bm{I}$. What this suggests is that in the presence of intermediate input linkages, what matters for the price index changes is not just how much each sector exports directly, $\bm{x}^T$, but also how much it exports indirectly through intermediate input linkages, $\bm{x}^T\bm{\Psi}_D$ \citep[see][]{DKMT21}. This mechanism also affects how much each factor ends up being exported and how much factor price changes are passed through to the CPI, since $\bm{x}^T\bm{\Psi}_D\bm{A}$ represents the factor content of exports when there are intermediate input linkages across sectors and -$\left(\bar{\bm{\lambda}}^T - \bar{\bm{x}}^T \bm{\Psi}_D\right)$ is the relevant ``Domar weight" for the pass-through of sectoral technology shocks to inflation. 


\paragraph{Illustration 3: Import price changes with intersectoral linkages and the network-adjusted import consumption share.}
This example illustrates that intersectoral linkages amplify the influence of import price changes on inflation. In the presence of intermediate input linkages and imported intermediate inputs, the \emph{direct} import consumption shares $\bar{\bm{b}}_M^T$ are \emph{not} a sufficient statistic for the effect of import prices on the CPI. To see this, fix factor prices and assume no productivity shocks, $\widehat{\bm{W}} = \bm{0}_F$ and  $\widehat{\bm{Z}} = \bm{0}_N$. Then 
\begin{align*}
    \widehat{P} & = \underbrace{\left(\bar{\bm{b}}_M^T+ \bar{\bm{b}}_D^T\bm{\Psi}_D \bm{\Gamma}\right)}_{\text{Network-adjusted import consumption share}}\widehat{\bm{P}}_M
\end{align*}
As was the case for the factor content of exports, this equation shows the importance of \emph{network-adjusted import consumption shares}. While domestic consumers purchase imports directly as final consumption ($\bar{\bm{b}}_M$), they also consume imports indirectly by purchasing domestically produced goods that directly or indirectly use imported intermediate inputs. This channel is captured by the second term on the right-hand side, $\bar{\bm{b}}_D^T\bm{\Psi}_D \bm{\Gamma}$, which captures the total import content of each domestically produced good when we account for intermediate input linkages. Intuitively, a rise in the price of import good  $m$ raises the marginal cost of a given producer $h$ by $\Gamma_{hm}$. This rise in the marginal cost implies that $P_h$ raises. This increase in $P_h$, through intermediate input linkages, raises the price of (say) good $i$, by $\Psi_{ih}$, which denotes the exposure of sector $i$ to changes in the price of sector $h$ after taking into account intermediate input linkages. This increase in the price of good $i$, in turn, is passed through to the consumer price index via $\bar{b}_i$.

\paragraph{Additional Models. } I provide two additional models in Appendix \ref{app:two_period} and \ref{appendix:dynamic_model}. In Appendix \ref{app:two_period}, I provide a detailed two-period model of a small open economy to show that the simplified model presented here shares the same intuition. The key idea follows \cite{BF22}, who in turn build on \cite{Krugman98} and \cite{EK12}, where we can separate a dynamic problem into two sub-periods: the present and the future. All action happens in the present, while the future can be taken as given. Shocks occur during the present and last only for that period, wherein the ``future", the economy returns to its initial no-shock equilibrium.  Conditional on this interpretation, the model in this section is isomorphic to a multi-period model. 

In Appendix \ref{appendix:dynamic_model},  I provide a three-sector (exportable, importable, and non-tradable) canonical small open economy dynamic model, as in Chapter 8 of \cite{SGU17Book}. There,  I embed a production network structure and show that the results also hold in that environment. In working out these additional models, I contribute to the literature by effectively embedding production networks into a small open economy setup and studying its consequences for inflation. 

\subsection{An alternative representation of factor markets: from factor prices to factor supplies}  

Factor price changes on the right-hand side of the equation (\ref{eq:main}) are exogenous and thus can be considered primitives in my exercise. However,  typical neoclassical models treat factor prices as endogenous and factor supply as exogenous. Writing the problem considering factor prices as given simplified the intuition for the main result of this paper. It also allowed me to differentiate the proximate causes of inflation between my model and a closed economy with or without production networks.\footnote{This is also the route followed by \cite{BF19BBVA} when studying aggregation in disaggregated economies via \emph{aggregate cost functions} rather than aggregate production functions.} I now show that the same intuition holds if we reverse the ordering, treating factor prices as endogenous outcomes and factor supply as exogenous objects.


\subsubsection{Solving in terms of factor supply quantities}

The key difference when solving for factor prices as endogenous objects is that we need to introduce factor market clearing conditions introduced below:
\begin{align*}
    \sum\limits_{i\in N}L_{if}&= \bar{L}_f \quad \forall f \in F,
\end{align*}
where the left-hand side is factor demand, and the right-hand side is factor supply. In what follows, I assume that factor supplies, $\bar{L}_f$, are exogenous. 

Recall that the expenditure-based share of factor $f$ can be written as
\begin{align*}
    \bar{\Lambda}_f &= \frac{W_f \bar{L}_f}{ \mathcal{M}},
\end{align*}
where I have imposed the factor market clearing condition and the cash-in-advance constraint.

Thus, changes in factor prices can be written as 
\begin{align}
    \widehat{W}_f &= \widehat{\bar{\Lambda}}_f + \widehat{\mathcal{M}} - \widehat{\bar{L}}_f,
\end{align}
which in vector form is 
\begin{align*}
    \widehat{\bm{W}} &= \widehat{\bm{\bar{\Lambda}}} + \bm{1}_F \widehat{\mathcal{M}} - \widehat{\bm{\bar{L}}}.
\end{align*}
Intuitively, factor prices can go up because (i)  demand is \emph{reallocated} toward that factor, as captured by $\widehat{\bm{\bar{\Lambda}}}$; (ii) aggregate demand is going up ($\widehat{\mathcal{M}})$; and (iii) there is a decrease in  (inelastic) factor supply ($\widehat{\bm{\bar{L}}}$). As shown in the proposition below, this decomposition allows me to write changes in the price index as a function of sectoral and aggregate shocks and also changes in these expenditure-based factor shares.

\begin{proposition} \label{prop3}
Consider a perturbation $(\widehat{\mathcal{M}}, \mathrm{d}T, \widehat{\bm{Z}}, \widehat{\bm{P}}_M, \widehat{\bm{X}}, \widehat{\bm{\bar{L}}})$ around some initial equilibrium. Up to a first order, changes in the aggregate price index, $\widehat{P}$, satisfy
\begin{align}
    \widehat{P} &= -\left(\bm{\bar{\lambda}}^T - \tilde{\bm{\lambda}}^T\right)\widehat{\bm{Z}} - \underbrace{\tilde{\bm{\Lambda}}^T\widehat{\bm{\bar{\Lambda}}} - \left(\bar{\bm{\Lambda}}^T - \bm{\tilde{\bm{\Lambda}}}^T\right)\widehat{\bm{\bar{L}}} + \frac{\mathrm{d}T}{\mathcal{M}} + \left(1 - \tilde{\bm{\Lambda}}^T\bm{1}_F\right) \widehat{\mathcal{M}}}_{\text{Factor Price Changes}}\nonumber \\ 
    & +  \left( \bm{\bar{b}}_M^T + \tilde{\bm{b}}_M^T\right)\widehat{\bm{P}}_M, \label{eq:main_3}
\end{align}
\end{proposition}
\begin{proof}
See Appendix \ref{proof:prop3}.
\end{proof}

\Cref{prop3} is an \emph{ex-post} sufficient statistics results in the spirit of \cite{BF22trade} since there is still one endogenous vector that needs to be solved for, namely $\widehat{\bm{\bar{\Lambda}}}$. Conditional on knowing this vector, we can compute the response of the CPI to changes in the other primitives. Note that in a closed economy, this term would be zero, since in this case $\tilde{\bm{\Lambda}} = \bm{0}_F$ and thus factor share reallocation would not have any first-order effect on inflation. 

Note that the only difference relative to the model with exogenous factor price changes is that we are now mapping factor price changes to other exogenous objects ($\widehat{\mathcal{M}}, \mathrm{d}T, \widehat{\bar{\bm{L}}}$). As in \cite{BF22}, decreases in factor supply are inflationary because they increase factor prices conditional on factor demands. The impact on inflation of money supply changes, $\widehat{\mathcal{M}}$, is dampened relative to the closed economy because factor prices have less pass through to inflation.  Increases in net transfers to the rest of the world, $\mathrm{d}T$, also increase CPI inflation because, conditional on money supply, they increase nominal GDP and thus increase factor prices. In this sense, factor price changes are ``sufficient statistics" for how money supply and net transfer changes affect the CPI.

A few additional remarks regarding Proposition \ref{prop3} are in order. First, note that sectoral export demands, $\bm{X}$, do not appear directly in this equation. It means that it can only affect inflation through its effect on $\widehat{\bm{\bar{\Lambda}}}$. Second,  as I show in Appendix \ref{app:factor_share_solving}, $\widehat{\bm{\bar{\Lambda}}}$ can be found by solving a linear system of equations. This system of equations depends on primitives, the production network structure, and the elasticities of substitution for producers and consumers. Thus, solving for $\widehat{\bm{\bar{\Lambda}}}$ requires us to take a stand on the values of elasticities of substitution of producers across different inputs and of consumers across different goods. Perhaps more important than this is that through this endogenous vector, elasticities of substitution matter to a first-order for CPI inflation in small open economies. Hence, even with a simple, sufficient statistics framework, elasticities of substitution matter to a first-order for inflation in small open economies, a result that contrasts with the closed-economy benchmark.

\section{The empirical relevance of adjustments of CPI elasticity}\label{sec: empirics}

This section examines the quantitative relevance of the proposed production network adjustments for inflation across small open economies. I start by describing the data sources and how I classify countries as small open economies. I then present three different results. First, I focus on the network-adjusted domestic consumption shares, which are the relevant elasticities for the pass-through of sectoral technology shocks to inflation. Second, I examine the adjustments to labor shares once we account for indirect exports. Finally, I compare direct and network-adjusted import consumption shares.

\subsection{Data} \label{subsec:data}

Although network-adjusted shares require more information than sales and factor shares, they are still easily computable from available data. In this subsection, I briefly describe the necessary data to compute them.


\paragraph{Input-Output Tables.} I calculate the objects $(\bm{\Omega}, \bm{\bar{x}}, \bm{\bar{b}}_D, \bm{\bar{b}}_M, \bm{\bar{\lambda}})$ using domestic input-output tables from the World Input-Output database release 2016, the latest available.

\paragraph{Penn World Tables (PWT).} I use version PWT 10.01. This dataset contains income, input, output, and productivity information between 1950 and 2019 for 183 countries.  This database is freely available to download at \href{https://www.rug.nl/ggdc/productivity/pwt/?lang=en}{https://www.rug.nl/ggdc/productivity/pwt/?lang=en}. 

Using this database, I construct two measures. First, I denote the share of world GDP accounted for by country $c$ as $\alpha_c$. Formally, 
\begin{align*}
    \alpha_c &= \frac{nGDP_c}{nGDP_W}, \quad nGDP_W = \sum\limits_{c \in C} nGDP_c
\end{align*}
I measure $nGDP_c$ using the series \emph{cgdpo}, which corresponds to the Output-side real GDP at current PPP's (in 2017 US\$ millions).

To measure trade openness, I use the series $csh\_x$ and $csh\_m$ in the PWT. The first corresponds to the ratio of merchandise exports over nominal GDP, while the second corresponds to imports over nominal GDP at current PPP's. I define the trade openness of country $c$ as 
\begin{align*}
    \text{Openness}_c&= \frac{\text{Exports}_c + \text{Imports}_c}{nGDP_c} = csh\_x_c - csh\_m_c,
\end{align*}
where the last line follows since in the data $csh\_m_c$ = $- \frac{\text{Imports}_c}{nGDP_c}$

\paragraph{Classifying Small Open Economies.} I apply two criteria to separate countries into small and non-small open economies according to the data. First, an economy is \emph{small} if $\alpha_c\leq 0.05$. Second, an economy is \emph{open} if $\text{Openness}_c\geq 0.3$. A country is a small open economy if it satisfies both conditions.

\subsection{Results}
In this subsection, I compare the network-adjusted objects with their closed economy and no-network adjustment counterparts whenever possible. All cross-sectional plots are based on the year 2014 unless stated otherwise.

\subsubsection{Network-adjusted domestic consumption shares}

I start by showing results for network-adjusted domestic consumption shares $\bar{\bm{\lambda}} - \bm{\bar{x}}^T \bm{\Psi}_D$. Figure \ref{fig:adj_domar_weights} shows three scenarios for the average sector in small open economies in panel (a) and for non-small open economies in panel (b). The x-axis shows the unadjusted Domar weights, while the y-axis shows the adjusted objects. Light squares in these figures refer to the export-adjusted weights, while dark points are export-network-adjusted weights. Thus, the dark points in these plots are the network-adjusted domestic consumption shares. As we can see, the adjustments are stronger for small open economies than non-small open economies. Moreover, we can see that the average Domar weight in small open economies is around 4 percent; it decreases to around 2.84 percent with the export adjustment and to around 2.31 percent when adjusting for network-adjusted exports. This is a non-negligible change. It suggests that the inflation impact of a given sized sectoral productivity shock in an average-sized sector will be dampened by around 50 percent for the average small open economy relative to the closed economy benchmark. 

\begin{figure}[htbp!]
    \centering
        \caption{Export and Network-Export adjusted Domar weights. }
    \label{fig:adj_domar_weights}
    \begin{minipage}{\textwidth}
    \caption*{\footnotesize \textbf{(a) Small Open Economies}}
    \centering
    \includegraphics[width = .8\textwidth]{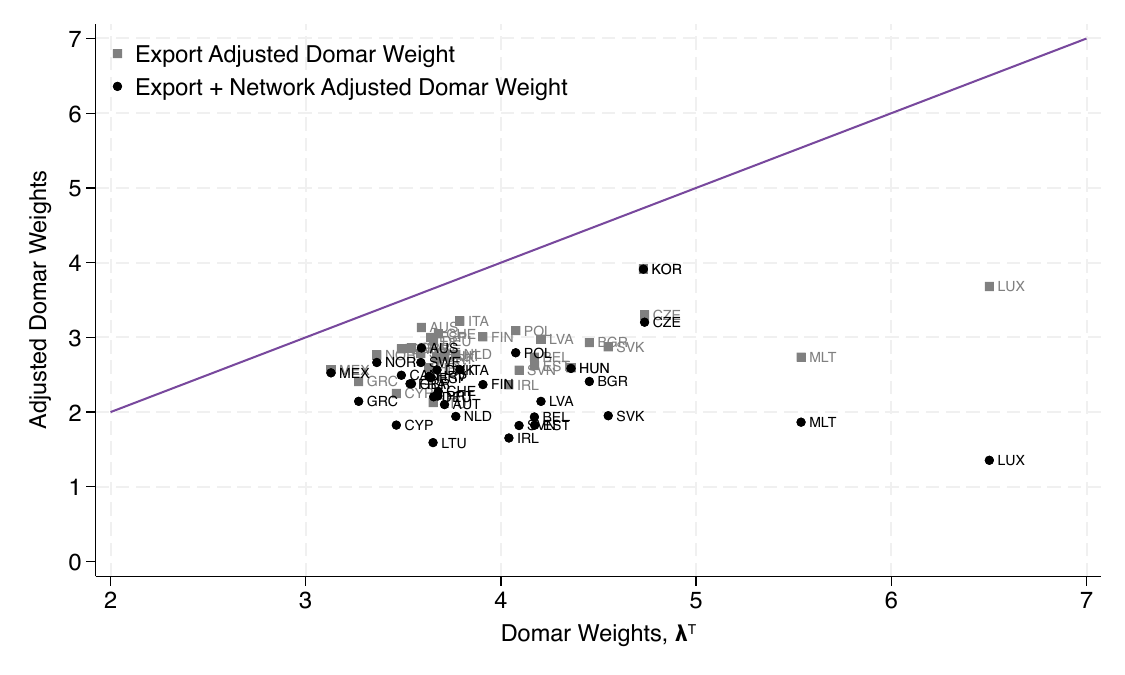}
    \end{minipage}
       \begin{minipage}{\textwidth}
        \centering
       \caption*{\footnotesize \textbf{(b) Non-Small Open Economies}}
\includegraphics[width = .8\textwidth]{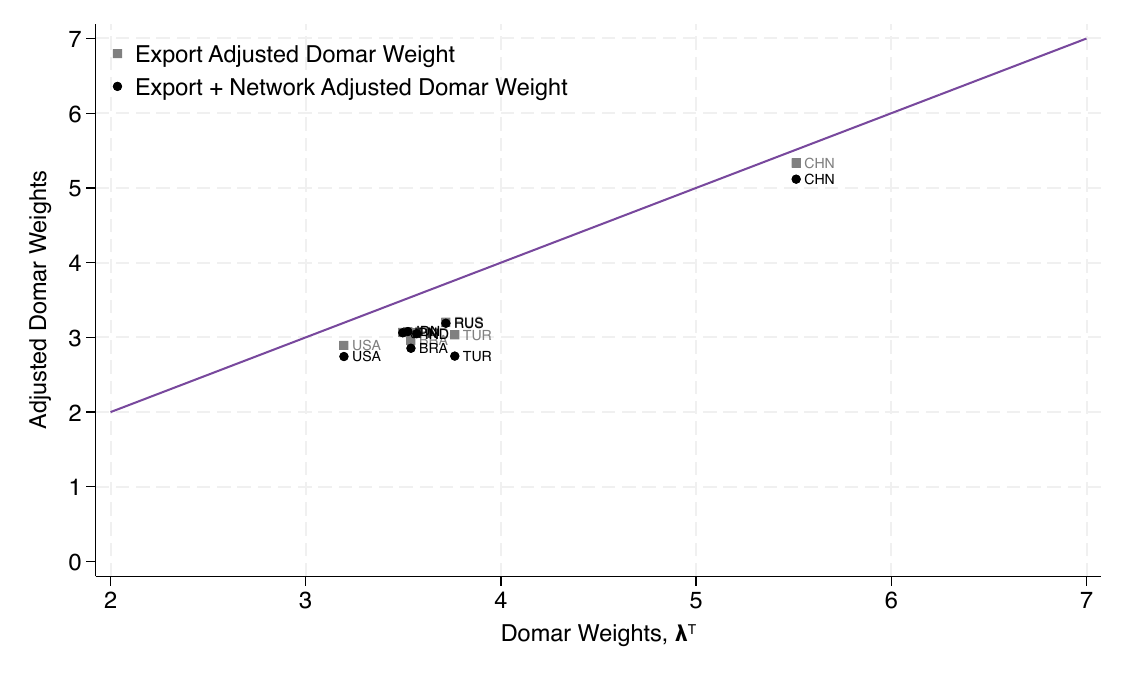}
    \end{minipage}
    \Fignote{This figure shows the average Domar weight for each country. The x-axis corresponds to the average Domar weight computed as in the closed economy model, $\bm{\lambda}^T$. The gray squares subtract only for direct exports i.e. $\bm{\bar{\lambda}}^T - \bm{\bar{x}}^T$. The black circles  further consider the production network structure, $\bm{\bar{\lambda}}^T - \bm{\bar{x}}^T\bm{\Psi}_D$. Panel (a) shows the results for small open economies, while Panel (b) shows the results for non-small open economies.}
\end{figure}

To provide a more concrete example, Figure \ref{fig:dw_uk} shows the three sectors for which network-export adjustment is the largest in the United Kingdom: administrative support, legal and accounting, and electricity, gas, and water. The latter sector is illustrative. Its Domar weight is around 5.95 percent. This number goes down only to 5.9 percent when we subtract direct exports. For all practical purposes, this means this sector is non-tradable. Once we consider indirect exports, however, the network-adjusted consumption share decreases to 4.4 percent. This illustrates how indirect linkages are quantitatively relevant and provide information beyond the direct export share.

\begin{figure}[htbp!]
    \caption{Three sectors with largest adjustments: United Kingdom.} \label{fig:dw_uk}
            \centering \includegraphics[width = 0.8\textwidth] {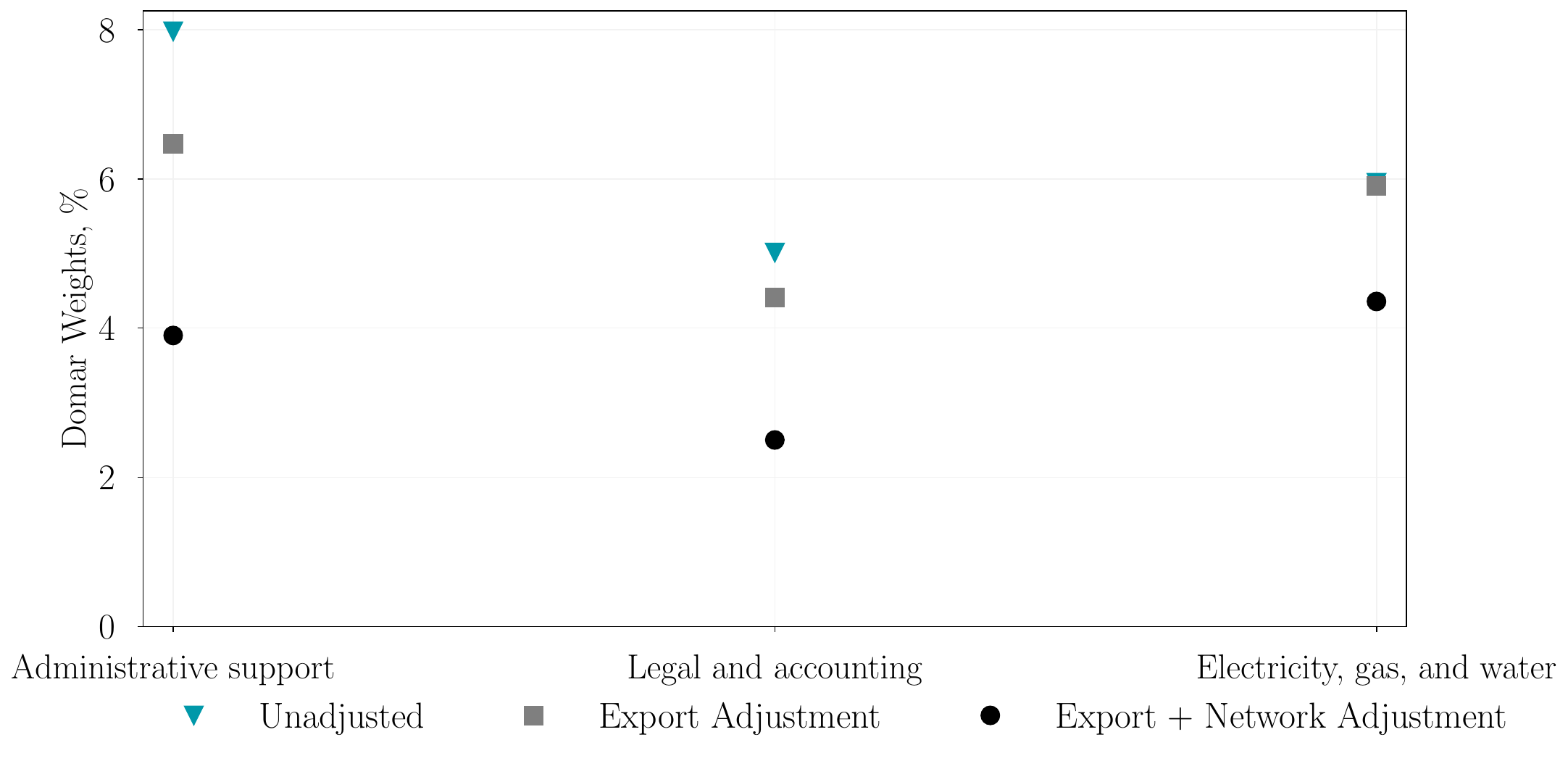}
            \Fignote{This figure shows the three sectors with the largest export network-adjusted share for the United Kingdom. }
    \end{figure}

\paragraph{Regression framework.} Which sectors and countries are most affected by network adjustments? To answer this question, I estimate the following cross-sectional regression
\begin{align}
    y_{sc}&= \alpha_s + \alpha_c + \varepsilon_{cs}, \label{eq:regression_framework}
\end{align}
where $y_{sc}$ represents the difference between a measure for the small open economy with a production network relative to the small open economy without networks for a given country $c$ and sector $s$. $\alpha_s$ is a sector-specific fixed effect, $\alpha_c$ is a country-specific fixed effect, and $\varepsilon_{cs}$ is an error term. From this regression, I get estimates of sector and country-specific fixed effects. Notice that these are identified up to a normalization, which in my case is that $\sum\limits_{s\in S}\hat{\alpha}_s = 0$ and $\sum\limits_{c\in C}\hat{\alpha}_c = 0$. All fixed effects are interpreted as deviations from the average fixed effects. 

In Panel (a) of Figure \ref{fig:reg_adj_domar_weights}, I show the country-fixed effect estimates when the left-hand side variable is the difference between the network-adjusted export share and the direct export share. Note that these country-fixed estimates tell us the average difference between these shares --- as a fraction of aggregate expenditure --- across sectors within a country. We can see that the countries with the largest adjustments are Luxembourg, Slovak Republic, Malta, and Latvia, while countries with the smaller adjustments are Korea, Hungary, and Mexico. These numbers indicate that the export sectors of the latter economies do not rely much on inputs from the domestic economy, not exporting much indirectly.\footnote{Remember that it does not need to be confused with a country that does not export at all.}

Panel (b) does the same exercise for the sector-fixed effect estimates. These estimates tell us the average difference between shares across countries within a sector. Financial services is the sector with the largest average production network adjustment. Note that the Electricity, Gas, and Water (EGSA) is the seventh sector with the largest difference, while Legal and accounting is the sixth sector. Thus, the examples cited above are not specific to the United Kingdom but have consistently large network adjustments across countries.  Intuitively, these sectors are important suppliers for domestic sectors that export directly or indirectly. 


These exercises suggest that accounting for the production network is important in computing inflation elasticities and that the adjustment varies substantially across countries and sectors.

\begin{figure}[htbp!]
    \centering
        \caption{Country and sector fixed effects: export-network adjusted - export adjusted.}
    \label{fig:reg_adj_domar_weights}
    \begin{minipage}{0.9\textwidth}
    \caption*{\footnotesize \textbf{(a) Country Fixed Effects}}
    \includegraphics[width = \textwidth]{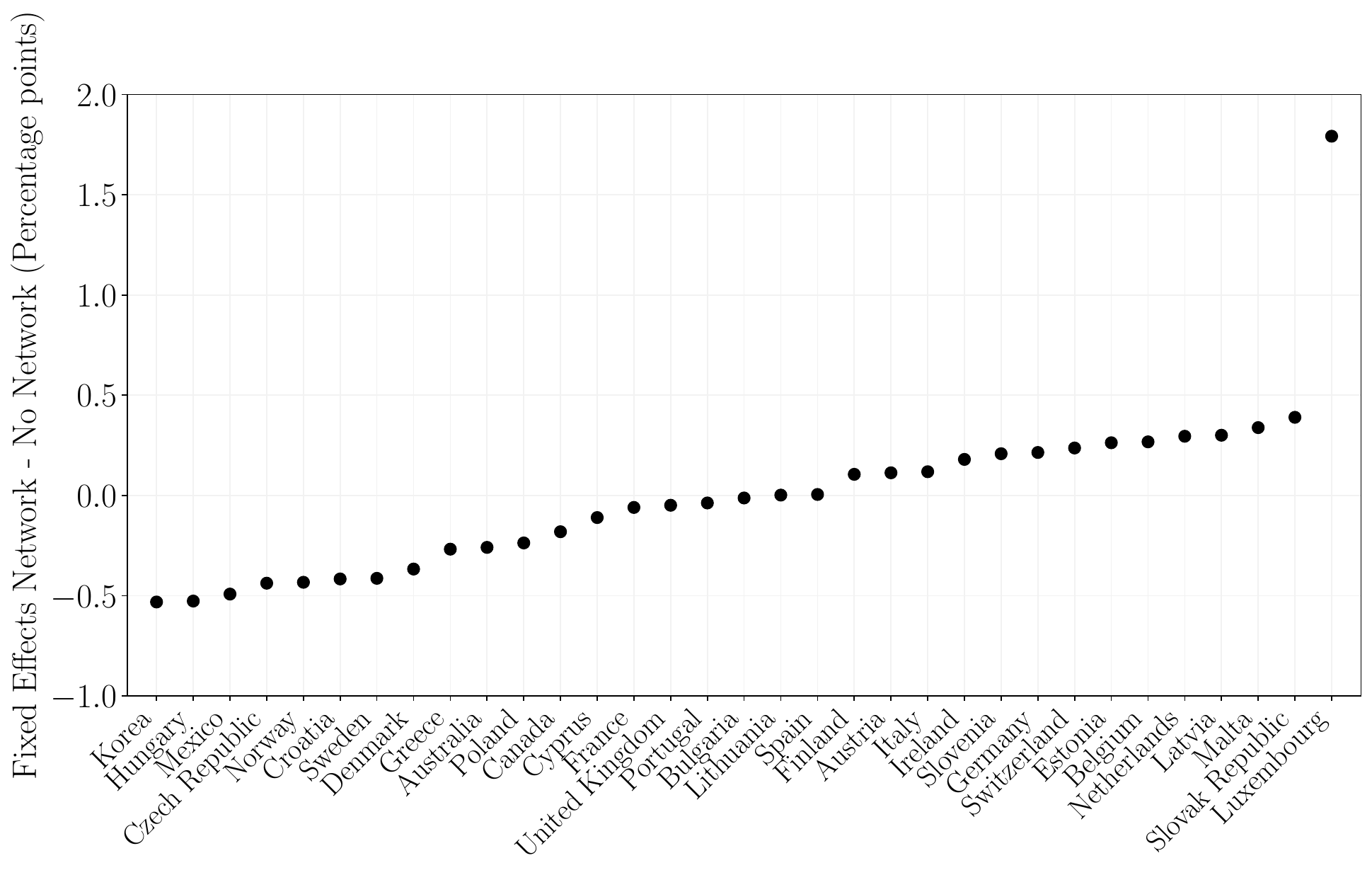}
    \end{minipage}
       \begin{minipage}{0.9\textwidth}
       \caption*{\footnotesize \textbf{(b) Sector Fixed Effects}}
\includegraphics[width = \textwidth]{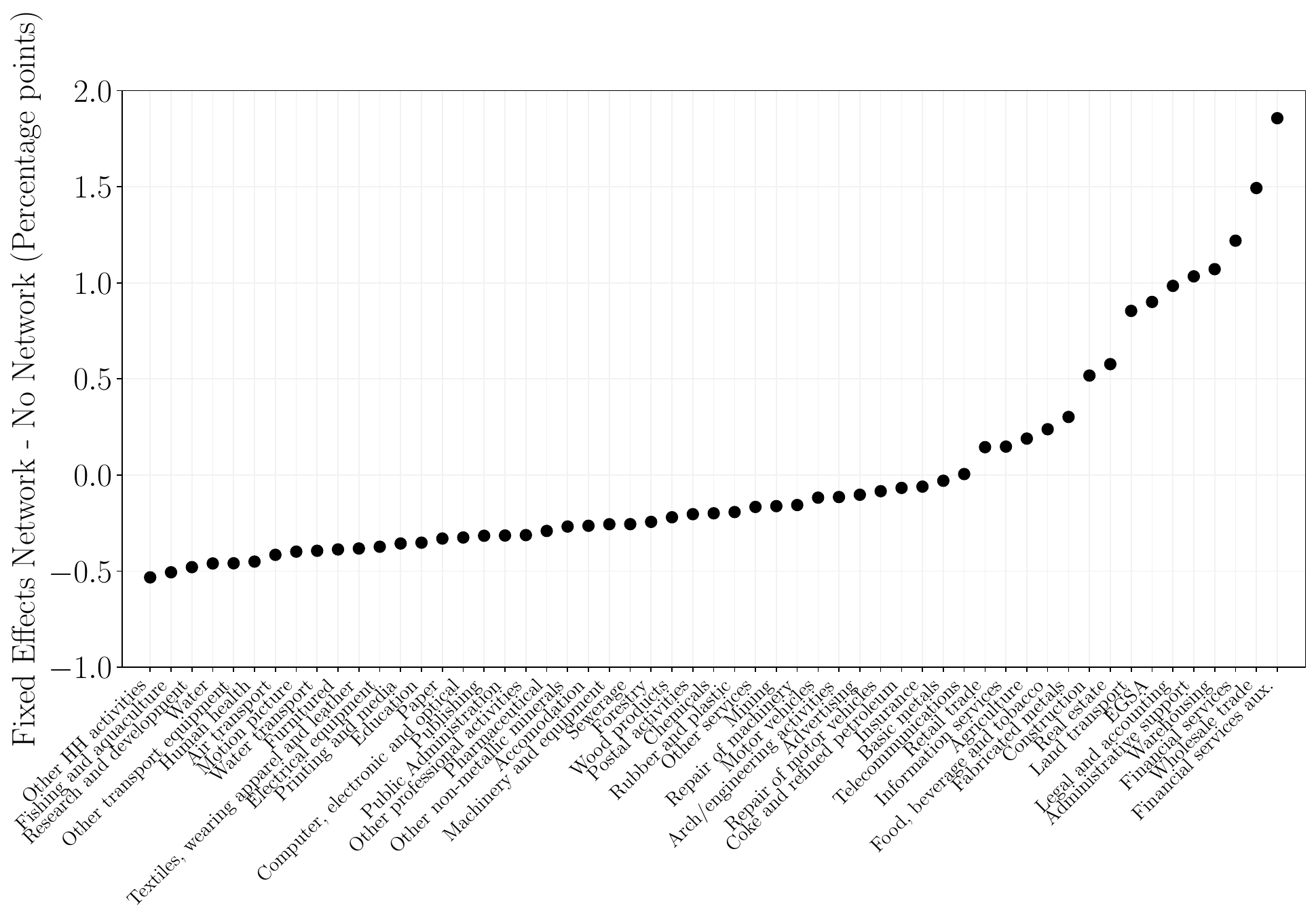}
    \end{minipage}
    \Fignote{This figure shows fixed effects from estimating equation (\ref{eq:regression_framework}) where the dependent variable is the difference between the network-adjusted export share and the direct export share. Panel (a) shows country fixed effects estimates, while Panel (b) shows sector fixed effects.}
\end{figure}

\subsubsection{Network-adjusted domestic factor demand}
I now conduct a similar exercise to examine the importance of network adjustments for factor shares. First, I study how the aggregate labor share in different countries varies depending on the export and network export adjustment. I then consider how sector-specific labor shares vary when considering direct and indirect exports.

\paragraph{Labor share.} Figure \ref{fig:labor_shares} shows the labor share for different economies on the x-axis and the network export-adjusted labor share on the y-axis. Black diamons shows small open economies and gray circles represents non-small open economies. As we can see, the adjustments are again significant for small open economies but not for non-small open economies. The average labor share across non-small open economies is 53 percent, while the network export-adjusted labor share is 50 percent, a negligible change. In contrast, the average labor share in small open economies is around 57 percent, while the network export-adjusted labor share is only 39 percent. This means the impact of a given wage increase will be 32 percent lower in a small open economy with production networks relative to an otherwise similar closed economy. 

\begin{figure}[t!]
    \centering 
    \caption{Labor share adjustments for different countries. } \label{fig:labor_shares}
        \includegraphics[scale = 0.8]{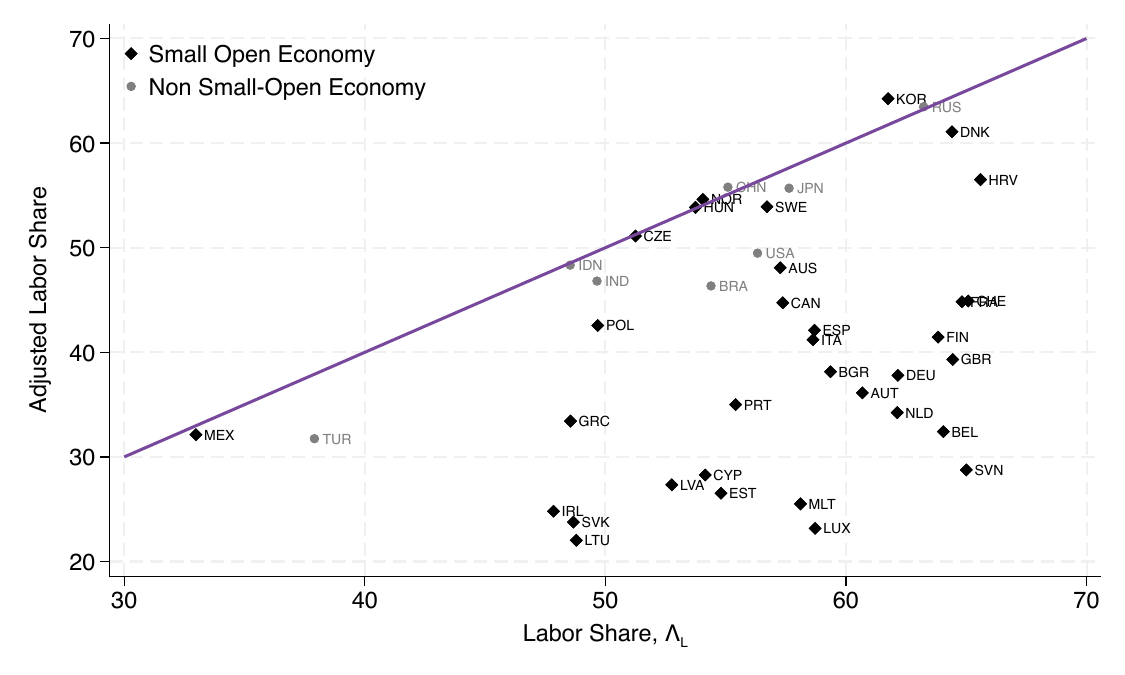}
        \Fignote{This figure shows the average labor share on the x-axis and the export-network adjusted labor share for small open economies in black diamonds and non-small open economies in gray circles.}
    \end{figure}

\paragraph{Sector-specific labor shares.} I now conduct a similar exercise to that of export and network-adjusted export measures above. Here, I consider the dependent variable to be the difference between the network-adjusted labor content of exports relative to the non-network-adjusted labor content of exports. 

This exercise illustrates the heterogeneity across sectoral labor markets. Before, I considered the aggregate labor share. However, this aggregate labor share is a weighted average of what happens at the sectoral level. It can be a misleading statistic for certain questions, especially in an environment such as COVID-19, where sectoral labor markets were hit differently.

Panel (a) of Figure \ref{fig:reg_adj_factor_shares} shows the results for the country-fixed effects, while Panel (b) shows the same but for sector-fixed effects. Apart from Luxembourg, the ranking differs from the network-adjusted domestic consumption share in Figure \ref{fig:reg_adj_domar_weights}. Interestingly, countries where sector-specific labor shares adjusted the most due to the domestic production network are the Netherlands, Slovenia, and Germany, while the ranking at the bottom stays the same. This says that Germany exhibits an average production network adjustment of sector-specific labor shares 0.15 percentage points larger than the adjustment for the average country.

Turning to the sector fixed-effects results, the sectors with the largest production network adjustment are Legal and Accounting, Wholesale Trade, and Administrative support. Legal and accounting, for example, has an average adjustment 0.6 percentage points larger than the average sector.  Since the 0.6 percentage point is an average across all countries, consider the Legal and Accounting sector in Germany as a concrete example. The share of this sector's labor on nominal GDP is around 2.6 percent of GDP. It goes down to 2.3 percent when we subtract exports and to 0.8 percent when we consider the domestic production network structure. Thus, ignoring the production network adjustment would significantly overstate how much wage changes sector pass-through to the CPI.

\begin{figure}[htbp!]
    \centering
        \caption{Country and sector fixed effects: export-network adjusted sector-specific factor shares. }
    \label{fig:reg_adj_factor_shares}
    \begin{minipage}{0.9\textwidth}
    \caption*{\footnotesize \textbf{(a) Country Fixed Effects}}
    \includegraphics[width = .95\textwidth]{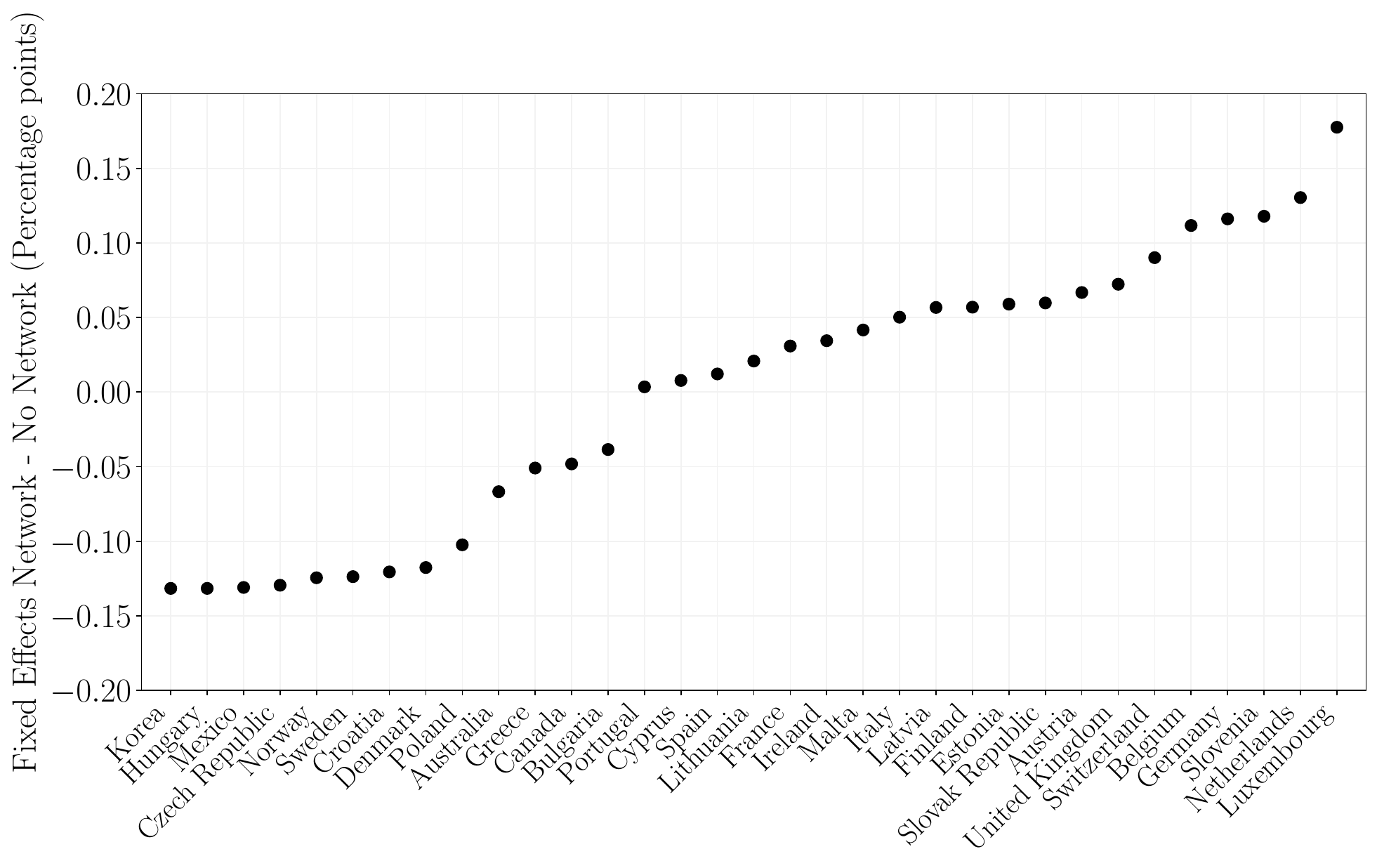}
    \end{minipage}
       \begin{minipage}{0.9\textwidth}
       \caption*{\footnotesize \textbf{(b) Sector Fixed Effects}}
\includegraphics[width = .95\textwidth]{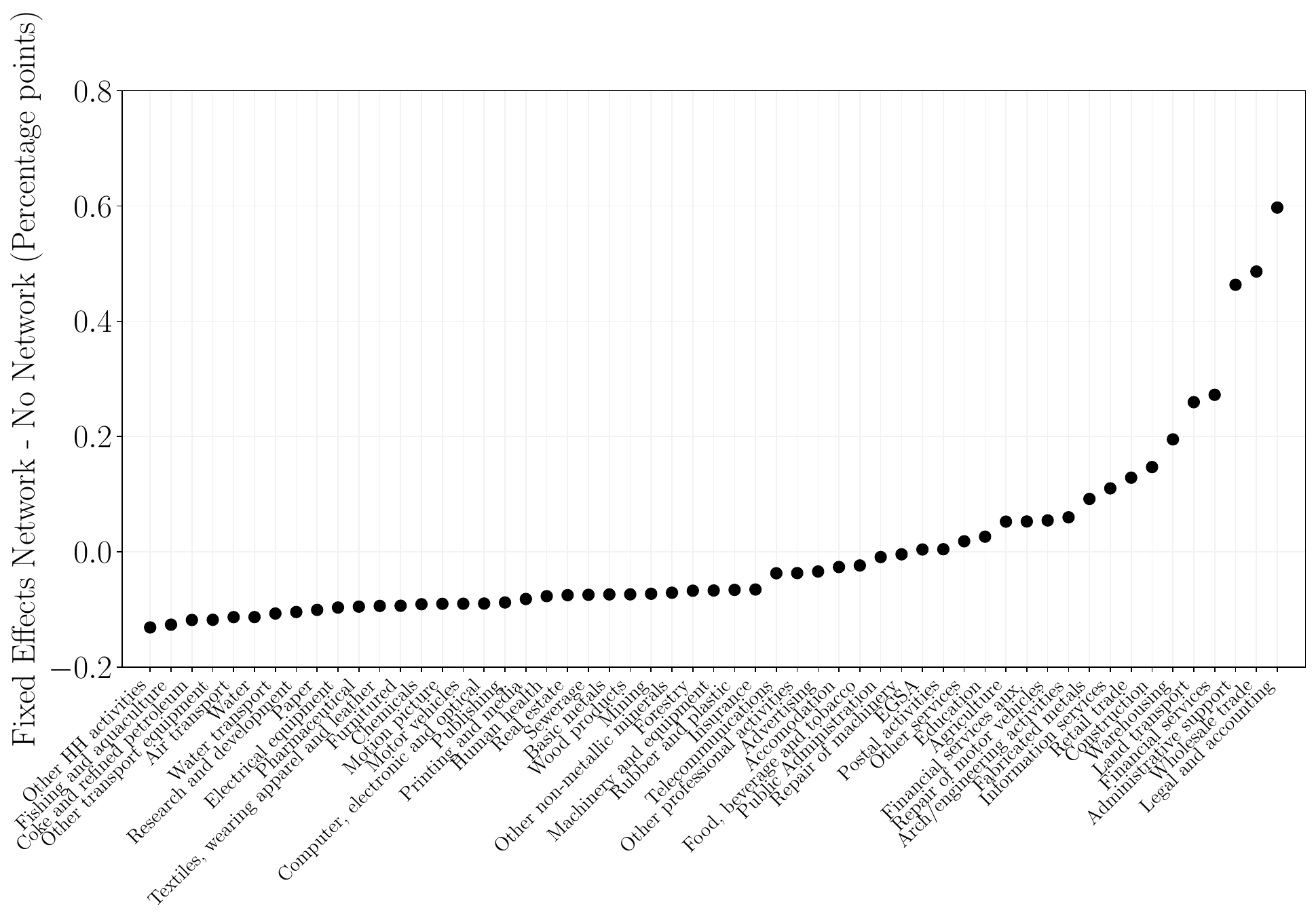}
    \end{minipage}
    \Fignote{This figure shows the fixed effects when the dependent variable is the difference between the network-adjusted sector-specific factor shares and the direct export share adjusted sector-specific factor shares. Panel (a) shows this difference for the country fixed effects estimates, while Panel (b) does the same for sector fixed effects.}
\end{figure}

\subsubsection{Network-adjusted import consumption shares}

As a final empirical exercise, I consider import consumption shares. Figure \ref{fig:import_consumption_shares}, provides a scatterplot of these shares across economies. On the x-axis, I show the direct import consumption share, while on the y-axis, I show the network-adjusted import consumption share. The average direct import consumption share across non-small open economies is 6.7 percent. It increases to 9.3 percent when considering the production network structure. While it increases by almost 3 percentage points, this change is small relative to the one I find for small open economies.
The average direct import consumption share across small open economies is around 17 percent. This number goes up to 30 percent when considering the production network structure. This is a 13 percentage points increase. It suggests that the pass-through from import price changes to inflation (almost) doubles when we introduce intersectoral linkages. 

\begin{figure}[t!]
    \centering
    \caption{Direct and Network-Adjusted import consumption shares.} \label{fig:import_consumption_shares}
            \includegraphics[scale = 0.8]{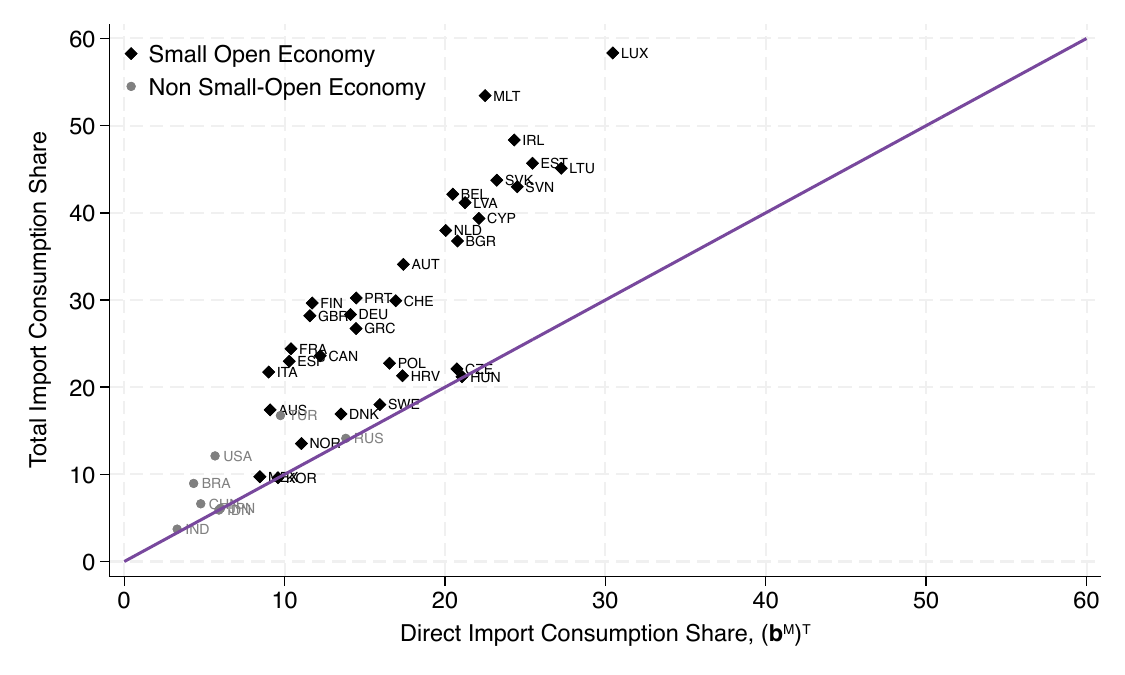}
            \Fignote{This figure shows the direct import consumption share on the x-axis and the network-adjusted import consumption share on the y-axis. Small open economies are the black diamonds, and non-small open economies are the gray circles.}
    \end{figure}

\section{The evolution of inflation in Chile and United Kingdom during COVID-19} \label{sec:application}

In this section, I use the model to study CPI inflation during the COVID-19 episode in Chile and the United Kingdom. 

This empirical examination requires more data relative to the previous section. While the earlier section showed information on the CPI elasticities and compared these across countries and sectors using the WIOT alone, this section requires taking a stand on the processes ($\widehat{\bm{W}}_t, \widehat{\bm{P}}_{Mt}, \widehat{\bm{Z}}_t$), which are not readily available for most countries worldwide. Therefore, I picked Chile and the United Kingdom, countries with all the necessary information to construct ($\widehat{\bm{W}}_t, \widehat{\bm{P}}_{Mt}, \widehat{\bm{Z}}_t$) that also belong to the small open economy category.

This is an ex-post exercise using existing data to analyze the past behavior of inflation between 2020 and 2022. Yet, Proposition \ref{prop1} is helpful for forecasting inflation and is thus a valuable tool for policymakers in small open economies. Provided that we have forecast information on the processes ($\widehat{\bm{W}}_t, \widehat{\bm{P}}_{Mt}, \widehat{\bm{Z}}_t$), we can combine this information with input-output tables to get an estimate of inflation. The accuracy of this exercise will depend on the accuracy of elasticities and that of the forecasted series. Throughout this section, I focus on the former, as it is the main point of this paper.

In what follows, I first describe the data. Then, I show how I map the model to the data. Finally, I discuss the results for Chile and the United Kingdom.


\subsection{Data}

\subsubsection{Chile}
\paragraph{Input-Output Tables.}
Since Chilean data is unavailable from the WIOT, I resort to Chilean National Accounts. I use \emph{Compilacion de Referencia} for year 2013. The structure is similar to that of the WIOT, having information on input-output linkages, final uses, and factor payments. Moreover, it is quite disaggregated, containing information for up to 171 industries. I collapse this data to a 17-sector classification due to data availability on sectoral wages.  This 17-sector classification is equivalent to SIC2.
\paragraph{Sectoral Productivity.} The ideal measure of productivity from the model is total factor productivity (TFP). However, TFP measures are hard to come by, especially at high frequencies and at the sectoral level. To circumvent this problem, I proxy sectoral TFP using sectoral labor productivity. I collect data on real GDP for the same 17 sectors and divide by total sectoral employment. Real GDP and sectoral employment data come from the Central Bank of Chile (CBCh) and are available quarterly from 1996 to 2022.
\paragraph{Sectoral Wages.} I source sectoral nominal wages from the Chilean National Institute of Statistics (INE) series \emph{Indice de Remuneraciones Nominal}. This database is available monthly from January 2016 to December 2022. To be consistent with the productivity data, I collapsed this data to a quarterly frequency.
\paragraph{Import Prices.} I use the import price index available from the CBCh quarterly from 2013 to 2022.

\subsubsection{United Kingdom}

\paragraph{Input-Output Tables.} I source data from the WIOT domestic tables as in the previous empirical section. I collapse these input-output tables into 20 industries to be consistent with the data on sectoral wages.
\paragraph{Sectoral Productivity.} I sourced data from the Office for National Statistics (ONS) of the United Kingdom. I downloaded quarterly estimates of labor productivity from the \emph{Flash productivity} report.\footnote{This data can be downloaded freely from the ``\emph{Flash productivity by section}" section at the ONS \href{https://www.ons.gov.uk/economy/economicoutputandproductivity/productivitymeasures/datasets/flashproductivitybysection}{here}.} This contains information for up to 17 industries. 
\paragraph{Sectoral Wages.} I source this data from the ONS of the United Kingdom. In particular, I use the dataset \texttt{EARN03}. This contains monthly information on average weekly earnings for around 20 industries. This dataset is available from 2000 to 2022.
\paragraph{Import Prices.} I use the import price index from the ONS (series \texttt{GD74}, dataset: \texttt{MM22}). This series is available at different frequencies. I use quarterly information from 2009 until 2022.

\subsection{Mapping the model to the data}

Before showing the results, a few remarks are in order. Since the model is static, all inherent inflation dynamics will combine the dynamics of exogenous variables and their interaction with the CPI elasticities. 

 First, I take all series and compute their level deviations from their value in 2018Q4. Formally, the sources of variation I feed in to construct implied inflation from the different models take the following form 
\begin{align*}
  \hat{y}_t &=   y_t - y_{2018Q4},
\end{align*}
where $y_t$ represents (the log) of any time series and $y_{2018Q4}$ is its value in 2018Q4. Notice that each vector now has a $t$ subscript as they are deviations from 2018Q4 at each time $t$.

In the above equation, I call the deviation $\hat{y}_t$ a ``shock". This differs from a structurally identified shock because I feed variation directly from the data, taking it as given.  With this caveat in mind and throughout this section, I refer to these $\hat{y}_t$ simply as shocks.

Using this procedure I construct counterparts to $\bm{\theta}_t = (\widehat{\bm{W}}_t, \widehat{\bm{P}}_{Mt}, \widehat{\bm{Z}}_t$) in the model. I measure factor prices as sectoral wages. I assume segmented labor markets such that there are different wages across sectors to capture better the behavior of labor markets during the COVID-19 episode, as highlighted in the recent literature \citep{BF22, dGKOSY22, dGKOSY23b,  dGKOSY23}. Since I cannot observe sector-specific prices for other factors, such as capital or land, I assume that those factor prices did not change over the sample period. 

CPI inflation in the data $\pi_t$, when $t$ refers to a quarter, is
\begin{align*}
  \pi_t &= \log P_t - \log P_{t-4}
\end{align*}
Combining the model and shocks, I have $\widehat{P}^{\text{Model}}_t$ as 
\begin{align*}
    \widehat{P}^{\text{Model}}_t &= -\sum\limits_{i\in N}\mathcal{R}^{CPI,Z}_i\widehat{Z}_{it} + \sum\limits_{f\in F} \mathcal{R}^{CPI,W}_f\widehat{W}_{ft} + \mathcal{R}^{CPI,M}_M\widehat{P}_{Mt}.
\end{align*}
Note that here $(\mathcal{R}^{CPI,Z}_i, \mathcal{R}^{CPI,W}_f, \mathcal{R}^{CPI,M}_M)$ stand for the responses of the CPI to changes in sectoral technology, factor prices, and import price, respectively. These objects are \emph{model-dependent} and thus will be different when considering the closed economy model, the small open economy model without production networks, and the small open economy with production networks. 

Finally, inflation \emph{from the model} is
\begin{align*}
    \pi_t^{\text{Model}}&= \widehat{P}^{\text{Model}}_t - \widehat{P}^{\text{Model}}_{t-4}.
\end{align*}

The approach of taking log differences relative to some initial point is the most transparent because it does not modify the data much, relative to other alternatives such as standard detrending procedures.

\subsection{Results}
In this subsection, I compare inflation implied by the models, $\pi_t^{\text{Model}}$, and that in the data.

Figures \ref{fig:chile_infl} and \ref{fig:uk_infl} show inflation in the data and the one implied by the model for Chile and the United Kingdom for 2020--2022, respectively. To highlight the distinct role of production networks and openness, I consider three models: a closed economy model (Closed, pink triangles), a small open economy without production networks (SOE no Network, green $\ast$), and a small open economy with production networks (SOE - Network, orange circles). I plot the model's numbers using symbols rather than lines to emphasize the absence of dynamics within the model apart from those generated by the shocks I am feeding in. 

Although the empirical exercise is fairly simple, it captures the data patterns well and more significantly for the small open economy with a production network in Chile and the United Kingdom. 

As pointed out, the model has no intrinsic dynamics: all the action over time comes from the dynamics in $\bm{\theta}_t$. A more meaningful comparison is to compare the moments implied by the model and those in the data. Table \ref{tab:moments_data_model} does precisely this and shows the first two moments of inflation in the data and the model. Panel (a) is for Chile, while Panel (b) shows the United Kingdom.

The average annual inflation in Chile between 2020 and 2022 was 6.13 percent, with a standard deviation of 3.89. The closed economy model delivers substantially lower mean inflation (0.98) and higher standard deviation (9.69) relative to the data. We can see that the sole introduction of a small open economy aspect, without production networks, gets us in the right direction as it exhibits a larger mean relative to the closed economy benchmark (1.45) and a lower standard deviation (6.88). The small open economy with production networks gets us closer to the data, with an average inflation of 2.41 and a standard deviation of 6.67. 

In the United Kingdom, the average inflation was 3.69 percent, almost half that of Chile during the same period, with a standard deviation of 3.11. The closed economy benchmark generates again too little average inflation (2.27 vs. 3.31) but now too low a standard deviation (2.57 vs. 3.11). As was the case for Chile, introducing the small open economy aspect put us in the right direction: inflation is higher on average (2.72) and has a higher standard deviation (2.64). Considering production networks again improves the results: the model exhibits an even higher mean (3.21) and standard deviation (3.00). 

\begin{figure}[htbp!]
    \centering
        \caption{Chile Inflation under different models.}
    \label{fig:chile_infl}
    \includegraphics[scale = 0.45]{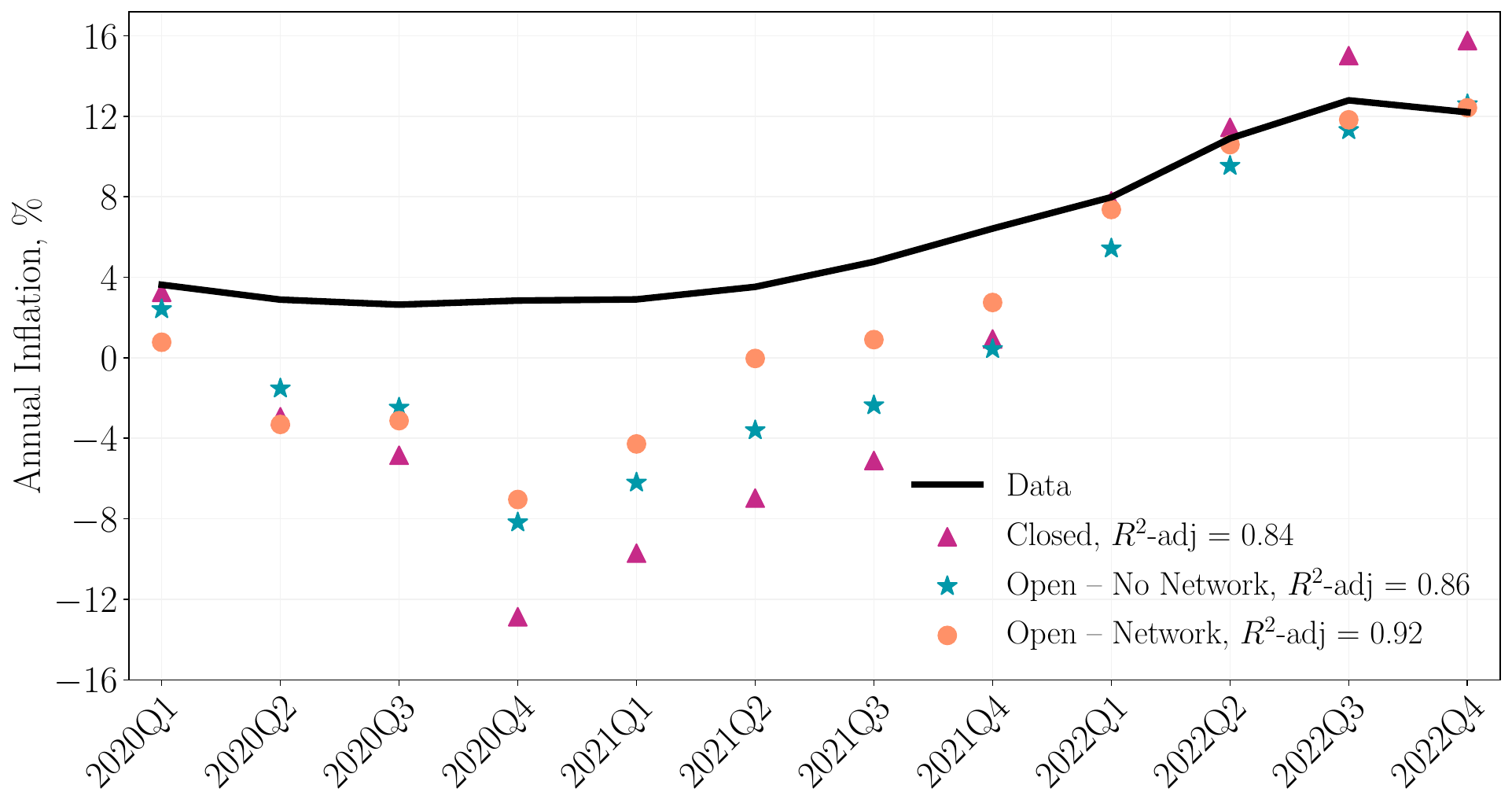}
    \Fignote{This figure shows inflation in the data and the one implied by the different models. The black line is the data. The pink triangles correspond to the closed economy model. The green $\ast$ are the small open economy model without production networks, and the orange circles correspond to the small open economy model with production networks.}
\end{figure}

\begin{figure}[htbp!]
    \centering
        \caption{United Kingdom Inflation under different models.}
    \label{fig:uk_infl}
    \includegraphics[scale = 0.45]{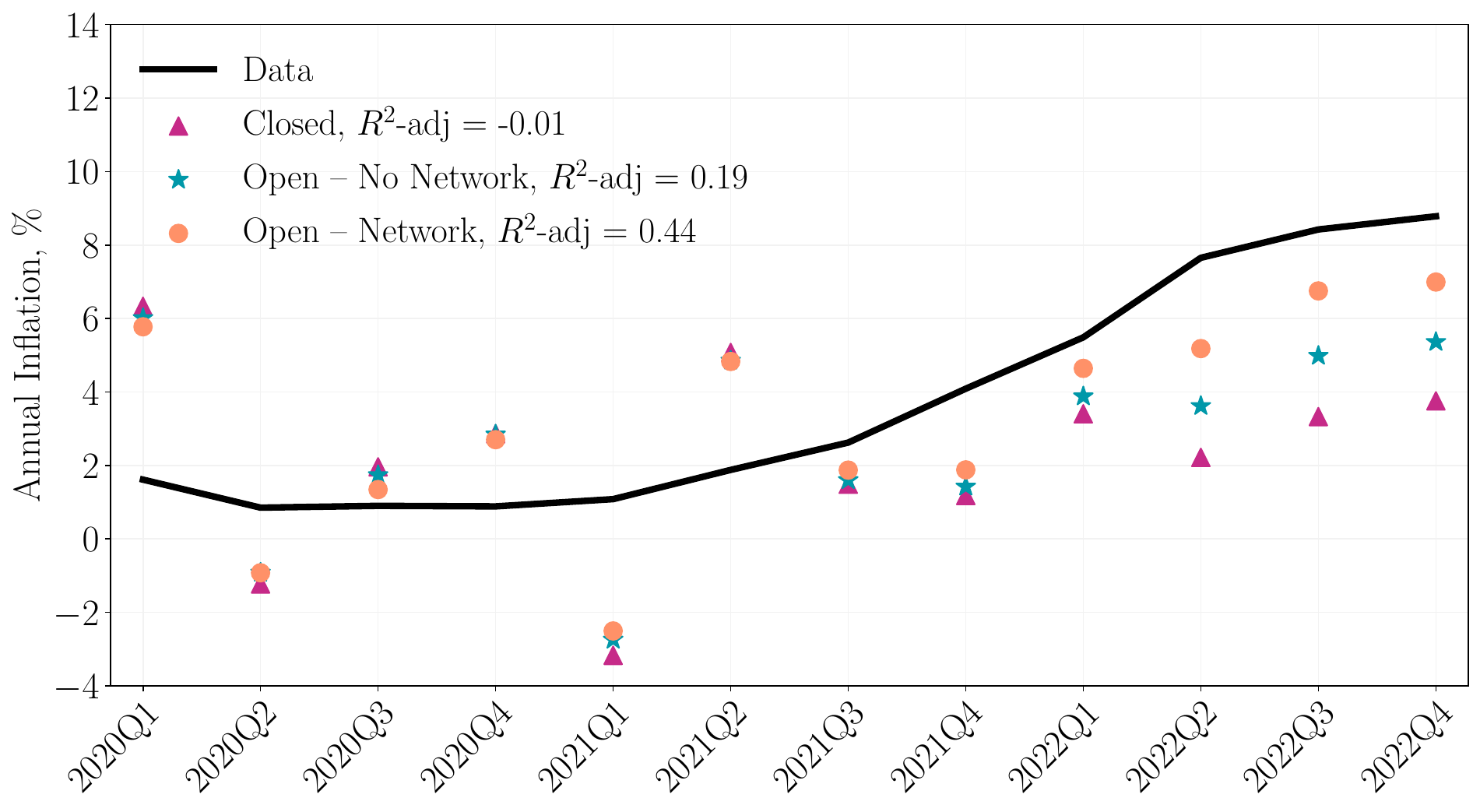}
    \Fignote{This figure shows inflation in the data and the one implied by the different models. The black line is the data. The pink triangles correspond to the closed economy model. The green $\ast$ are the small open economy model without production networks, and the orange circles correspond to the small open economy model with production networks.}
\end{figure}

To sum up, this exercise suggests that a small open economy model with production networks better matches inflation moments during 2020 -- 2022 than a closed economy model and a small open economy model without production networks. To be fair, the small open economy with production networks should indeed do better than the other two as it adds a piece of realism missing from these other models, namely, intersectoral linkages. The question is how much. The results here are suggestive evidence that this is quantitatively relevant. Of course, the stylized model has many missing parts, but remarkably, such a stylized exercise matches these inflation data moments well.

\begin{table}[htb!]
    \centering
        \caption{Average Inflation, 2020 -- 2022.}
    \label{tab:moments_data_model}
    \begin{tabular}{l c c c c}
    \toprule
    & \multicolumn{2}{c}{Panel (a): Chile} & \multicolumn{2}{c}{Panel (b): United Kingdom}\\
    \cmidrule{2-3} \cmidrule{4-5}
    & Mean & Std. Dev. & Mean & Std. Dev\\
    \cmidrule{2-5}
    Data & 6.13 & 3.89 & 3.69 & 3.11 \\
    \\
    \emph{Model} & \\
       Closed  & 0.98 & 9.69& 2.27 & 2.57\\
        SOE no Network & 1.45 & 6.88& 2.72 & 2.64\\
        SOE - Network & 2.41 & 6.67 & 3.21 & 3.00\\
        \bottomrule
    \end{tabular}
\Fignote{This table shows the mean and standard deviation of inflation for the data and the different models. The Closed model uses the implied elasticities as if we consider the economies as closed. The SOE no network model considers the elasticities in a small open economy that does not feature any production network. Finally, the SOE - Network model accounts for network linkages.}
\end{table}
\section{Conclusion}\label{sec:conclusion}
I study inflation in small open economies with production networks. Theoretically and empirically, I show that production networks matter for the effect of sectoral technology shocks, factor prices, and import prices on CPI inflation. I argue that the reason why the interaction of trade and production network matters is because opening up the economy is one of the ways to break the equivalence between what is produced within borders and what is consumed by the domestic consumer, for whom the CPI is relevant. Once we break that relationship, the production network amplifies this discrepancy via indirect linkages: Non-exporters become indirect exporters, while non-importers become indirect importers. This ultimately affects the exposure of the domestic consumer to a different set of shocks. The production network thus has a first order impact on inflation, with sales and factor shares no longer being sufficient statistics to study how changes in sectoral technology or factor prices pass through to inflation, as it would be the case in a closed economy. 

In a small open economy with production networks, indirect exporting dampens domestic shocks relative to an otherwise equivalent closed economy. Foreign shocks, such as import price shocks, are amplified relative to an otherwise equivalent small open economy without production networks. Which channels dominate at the aggregate level depends on the production network structure and is, in the end, a quantitative question. I show that the production network adjustments on both the export and import side matter quantitatively across a set of small open economies. I apply the small open economy model with production networks to the recent inflationary episode in Chile and the United Kingdom. Including production networks helps better match the mean and variance of inflation of these countries between 2020 and 2022. 


\newpage
\singlespacing
\let\normalsize\small
\small
\bibliography{references}

@article{BF22,
  title={Supply and demand in disaggregated Keynesian economies with an application to the Covid-19 crisis},
  author={Baqaee, David and Farhi, Emmanuel},
  journal={American Economic Review},
  volume={112},
  number={5},
  pages={1397--1436},
  year={2022}
}

@article{ACCDP22,
  title={Imports, exports, and earnings inequality: Measures of exposure and estimates of incidence},
  author={Adao, Rodrigo and Carrillo, Paul and Costinot, Arnaud and Donaldson, Dave and Pomeranz, Dina},
  journal={The Quarterly Journal of Economics},
  volume={137},
  number={3},
  pages={1553--1614},
  year={2022},
  publisher={Oxford University Press}
}

@article{Basu95,
  title={Intermediate Goods and Business Cycles: Implications for Productivity and Welfare},
  author={Susanto Basu},
  journal={The American Economic Review},
  volume={85},
  number={3},
  pages={512--531},
  year={1995},
  publisher={American Economic Association}
}

@article{DKMT21,
  title={Trade and domestic production networks},
  author={Dhyne, Emmanuel and Kikkawa, Ayumu Ken and Mogstad, Magne and Tintelnot, Felix},
  journal={The Review of Economic Studies},
  volume={88},
  number={2},
  pages={643--668},
  year={2021},
  publisher={Oxford University Press}
}

@article{GM05,
  title={Monetary policy and exchange rate volatility in a small open economy},
  author={Gali, Jordi and Monacelli, Tommaso},
  journal={The Review of Economic Studies},
  volume={72},
  number={3},
  pages={707--734},
  year={2005},
  publisher={Wiley-Blackwell}
}

@article{Rubbo23,
  title={Networks, Phillips curves, and monetary policy},
  author={Rubbo, Elisa},
  journal={Econometrica},
  volume={91},
  number={4},
  pages={1417--1455},
  year={2023},
  publisher={Wiley Online Library}
}

@article{LTS22,
  title={Optimal Monetary Policy in Production Networks},
  author={La'o, Jennifer and Tahbaz-Salehi, Alireza},
  journal={Econometrica},
  volume={90},
  number={3},
  pages={1295--1336},
  year={2022},
  publisher={Wiley Online Library}
}

@article{AB22,
  title={Inflation and GDP Dynamics in Production Networks: A Sufficient Statistics Approach},
  author={Afrouzi, Hassan and Bhattarai, Saroj},
  journal={Available at SSRN},
  year={2022}
}

@article{Munoz23,
  title={Trading Nontradables: The Implications of Europe’s Job-Posting Policy},
  author={Mu{\~n}oz, Mathilde},
  journal={The Quarterly Journal of Economics},
  pages={Forthcoming},
  year={2023},
  publisher={Oxford University Press}
}

@techreport{GG23,
  title={Oil Prices, Monetary Policy and Inflation Surges},
  author={Gagliardone, Luca and Gertler, Mark},
  year={2023},
  institution={National Bureau of Economic Research}
}

@article{LV23,
  title={Propagation of shocks in an input-output economy: Evidence from disaggregated prices},
  author={Luo, Shaowen and Villar, Daniel},
  journal={Journal of Monetary Economics},
  year={2023},
  publisher={Elsevier}
}

@techreport{CJ20,
  title={Offshoring and inflation},
  author={Comin, Diego A and Johnson, Robert C},
  year={2020},
  institution={National Bureau of Economic Research}
}

@article{HSS22,
  title={Multilateral Comovement in a New Keynesian World: A Little Trade Goes a Long Way},
  author={Ho, Paul and Sarte, Pierre-Daniel G and Schwartzman, Felipe F},
  year={2022},
  publisher={FRB Richmond Working Paper}
}

@techreport{dGKOSY22,
  title={Global supply chain pressures, international trade, and inflation},
  author={{di Giovanni}, Julian and Kalemli-{\"O}zcan, \c{S}ebnem and Silva, Alvaro and Y{\i}ld{\i}r{\i}m, Muhammed A},
  year={2022},
  institution={National Bureau of Economic Research}
}

@article{dGKOSY23,
Author = {{di Giovanni}, Julian and Kalemli-{\"O}zcan, \c{S}ebnem and Silva, Alvaro and Y{\i}ld{\i}r{\i}m, Muhammed A.},
Title = {Quantifying the Inflationary Impact of Fiscal Stimulus under Supply Constraints},
Journal = {AEA Papers and Proceedings},
Volume = {113},
Year = {2023},
Month = {May},
Pages = {76-80},
DOI = {10.1257/pandp.20231028},
URL = {https://www.aeaweb.org/articles?id=10.1257/pandp.20231028}}

@article{dGKOSY23b,
  title={Pandemic-Era Inflation Drivers and Global Spillovers},
  author={{di Giovanni}, Julian and Kalemli-{\"O}zcan, \c{S}ebnem and Silva, Alvaro and Y{\i}ld{\i}r{\i}m, Muhammed A},
  year={2023},
  journal={Working Paper}
}

@techreport{FGI22,
  title={The inflationary effects of sectoral reallocation},
  author={Ferrante, Francesco and Graves, Sebastian and Iacoviello, Matteo},
  year={2022},
  institution={Working paper}
}

@article{LW23,
  title={Wage Price Spirals},
  author={Lorenzoni, Guido and Werning, Ivan},
  year={2023},
  journal={Brookings Papers on Economic Activity}
}

@article{GLSW22,
  title={Macroeconomic implications of COVID-19: Can negative supply shocks cause demand shortages?},
  author={Guerrieri, Veronica and Lorenzoni, Guido and Straub, Ludwig and Werning, Iv{\'a}n},
  journal={American Economic Review},
  volume={112},
  number={5},
  pages={1437--74},
  year={2022}
}

@article{ALS19,
  title={International inflation spillovers through input linkages},
  author={Auer, Raphael A and Levchenko, Andrei A and Saur{\'e}, Philip},
  journal={Review of Economics and Statistics},
  volume={101},
  number={3},
  pages={507--521},
  year={2019},
  publisher={MIT Press One Rogers Street, Cambridge, MA 02142-1209, USA journals-info~…}
}

@article{BR22,
  title={Micro Propagation and Macro Aggregation},
  author={Baqaee, David and Rubbo, Elisa},
  year={2023},
  journal={Forthcoming, Annual Review of Economics}
}

@techreport{BF22trade,
  title={Networks, barriers, and trade},
  author={Baqaee, David and Farhi, Emmanuel},
  year={2022},
  institution={National Bureau of Economic Research}
}

@techreport{QWXZ24,
  title={Monetary Policy in Open Economies with Production Networks},
  author={Qiu, Zhesheng and Wang, Yicheng and Xu, Le and Zanetti, Francesco},
  year={2024},
  institution={Working Paper}
}

@article{Matsumara22,
title = {What price index should central banks target? An open economy analysis},
journal = {Journal of International Economics},
volume = {135},
pages = {103554},
year = {2022},
issn = {0022-1996},
doi = {https://doi.org/10.1016/j.jinteco.2021.103554},
url = {https://www.sciencedirect.com/science/article/pii/S0022199621001343},
author = {Misaki Matsumura}}

@techreport{BB23a,
  title={What Caused the US Pandemic-Era Inflation?},
  author={Blanchard, Olivier J and Bernanke, Ben S},
  year={2023},
  institution={National Bureau of Economic Research}
}

@techreport{CJJ23,
  title={Supply chain constraints and inflation},
  author={Comin, Diego A and Johnson, Robert C and Jones, Callum J},
  year={2023},
  institution={National Bureau of Economic Research}
}

@inproceedings{CP07,
  title={The Simple Geometry of Transmission and Stabilization in Closed and Open Economies},
  author={Corsetti, Giancarlo and Pesenti, Paolo},
  booktitle={NBER international seminar on macroeconomics},
  volume={2007},
  number={1},
  pages={65--129},
  year={2007},
  organization={The University of Chicago Press Chicago, IL}
}

@article{Huneeus18,
  title={Production network dynamics and the propagation of shocks},
  author={Huneeus, Federico},
  journal={Graduate thesis, Princeton University, Princeton, NJ},
  pages={52},
  year={2018}
}

@inproceedings{DKMT23,
  title={Measuring the share of imports in final consumption},
  author={Dhyne, Emmanuel and Kikkawa, Ayumu Ken and Mogstad, Magne and Tintelnot, Felix},
  booktitle={AEA Papers and Proceedings},
  volume={113},
  pages={81--86},
  year={2023},
  organization={American Economic Association 2014 Broadway, Suite 305, Nashville, TN 37203}
}

@article{BF20,
  title={Productivity and misallocation in general equilibrium},
  author={Baqaee, David Rezza and Farhi, Emmanuel},
  journal={The Quarterly Journal of Economics},
  volume={135},
  number={1},
  pages={105--163},
  year={2020},
  publisher={Oxford University Press}
}

@article{Woodford03,
  title={Interest and Prices: Foundations of a Theory of Monetary Policy},
  author={Woodford, Michael},
  year={2003},
  publisher={Princeton University Press}
}

@article{AAS23,
  title={The Aggregate Effects of Sectoral Shocks in an Open Economy},
  author={Andrade, Philippe and Arazi, Martin and Sheremirov, Viacheslav},
  year={2023},
  journal={FRB of Boston Working Paper}
}

@article{BFM23,
  title={A Fiscal Theory of Persistent Inflation},
  author={Bianchi, Francesco and Faccini, Renato and Melosi, Leonardo},
  journal={The Quarterly Journal of Economics},
  pages={qjad027},
  year={2023},
  publisher={Oxford University Press}
}

@article{BF19,
  title={The macroeconomic impact of microeconomic shocks: Beyond Hulten's theorem},
  author={Baqaee, David and Farhi, Emmanuel},
  journal={Econometrica},
  volume={87},
  number={4},
  pages={1155--1203},
  year={2019},
  publisher={Wiley Online Library}
}

@article{Krugman98,
  title={It's baaack: Japan's slump and the return of the liquidity trap},
  author={Krugman, Paul},
  journal={Brookings Papers on Economic Activity},
  volume={1998},
  number={2},
  pages={137--205},
  year={1998},
  publisher={JSTOR}
}

@article{EK12,
  title={Debt, deleveraging, and the liquidity trap: A Fisher-Minsky-Koo approach},
  author={Eggertsson, Gauti B and Krugman, Paul},
  journal={The Quarterly Journal of Economics},
  volume={127},
  number={3},
  pages={1469--1513},
  year={2012},
  publisher={MIT Press}
}

@article{Baqaee15,
  title={Targeted Fiscal Policy},
  author={Baqaee, David},
  journal={Working Paper},
  year={2015}
}

@article{Hulten78,
  title={Growth accounting with intermediate inputs},
  author={Hulten, Charles R},
  journal={The Review of Economic Studies},
  volume={45},
  number={3},
  pages={511--518},
  year={1978},
  publisher={Wiley-Blackwell}
}

@article{Svensson00,
  title={Open-economy inflation targeting},
  author={Svensson, Lars EO},
  journal={Journal of international economics},
  volume={50},
  number={1},
  pages={155--183},
  year={2000},
  publisher={Elsevier}
}

@article{ACD17,
  title={Nonparametric counterfactual predictions in neoclassical models of international trade},
  author={Adao, Rodrigo and Costinot, Arnaud and Donaldson, Dave},
  journal={American Economic Review},
  volume={107},
  number={3},
  pages={633--689},
  year={2017},
  publisher={American Economic Association 2014 Broadway, Suite 305, Nashville, TN 37203}
}

@article{BL20,
  title={Distortions in production networks},
  author={Bigio, Saki and La’o, Jennifer},
  journal={The Quarterly Journal of Economics},
  volume={135},
  number={4},
  pages={2187--2253},
  year={2020},
  publisher={Oxford University Press}
}

@article{CV99,
  title={Inflation stabilization and BOP crises in developing countries},
  author={Calvo, Guillermo A and V{\'e}gh, Carlos A},
  journal={Handbook of macroeconomics},
  volume={1},
  pages={1531--1614},
  year={1999},
  publisher={Elsevier}
}

@article{CV95,
  title={Fighting inflation with high interest rates: the small open economy case under flexible prices},
  author={Calvo, Guillermo A and V{\'e}gh, Carlos A},
  journal={Journal of Money, Credit and Banking},
  volume={27},
  number={1},
  pages={49--66},
  year={1995},
  publisher={JSTOR}
}

@article{CRV95,
  title={Targeting the real exchange rate: theory and evidence},
  author={Calvo, Guillermo A and Reinhart, Carmen M and Vegh, Carlos A},
  journal={Journal of development economics},
  volume={47},
  number={1},
  pages={97--133},
  year={1995},
  publisher={Elsevier}
}

@book{Vegh2013,
  title={Open economy macroeconomics in developing countries},
  author={V{\'e}gh, Carlos A},
  year={2013},
  publisher={MIT press}
}

@misc{MW23,
  title={Delayed Inflation in Supply Chains: Theory and Evidence},
  author={Minton, Robert and Wheaton, Brian},
  year={2023},
  publisher={Manuscript}
}

@incollection{CDL10,
  title={Optimal monetary policy in open economies},
  author={Corsetti, Giancarlo and Dedola, Luca and Leduc, Sylvain},
  booktitle={Handbook of monetary economics},
  volume={3},
  pages={861--933},
  year={2010},
  publisher={Elsevier}
}

@article{FM08,
  title={Optimal monetary policy in a small open economy with home bias},
  author={Faia, Ester and Monacelli, Tommaso},
  journal={Journal of Money, Credit and Banking},
  volume={40},
  number={4},
  pages={721--750},
  year={2008},
  publisher={Wiley Online Library}
}

@article{GL07,
  title={Menu costs and Phillips curves},
  author={Golosov, Mikhail and Lucas, Robert E},
  journal={Journal of Political Economy},
  volume={115},
  number={2},
  pages={171--199},
  year={2007},
  publisher={The University of Chicago Press}
}

@book{SGUW22,
  title={International Macroeconomics: A modern approach},
  author={Schmitt-Groh{\'e}, Stephanie and Uribe, Mart{\'\i}n and Woodford, Michael},
  year={2022},
  publisher={Princeton University Press}
}

@article{BF19BBVA,
  title={JEEA-FBBVA Lecture 2018: The Microeconomic Foundations of Aggregate Production Functions},
  author={Baqaee, David and Farhi, Emmanuel},
  journal={Journal of the European Economic Association},
  volume={17},
  number={5},
  pages={1337--1392},
  year={2019},
  publisher={Oxford University Press}
}

@article{BCORY13,
  title={Financial crises and macro-prudential policies},
  author={Benigno, Gianluca and Chen, Huigang and Otrok, Christopher and Rebucci, Alessandro and Young, Eric R},
  journal={Journal of International Economics},
  volume={89},
  number={2},
  pages={453--470},
  year={2013},
  publisher={Elsevier}
}

@article{SGU03,
  title={Closing small open economy models},
  author={Schmitt-Groh{\'e}, Stephanie and Uribe, Mart{\i}n},
  journal={Journal of International Economics},
  volume={61},
  number={1},
  pages={163--185},
  year={2003},
  publisher={Elsevier}
}

@book{SGU17Book,
  title={Open economy macroeconomics},
  author={Uribe, Martin and Schmitt-Groh{\'e}, Stephanie},
  year={2017},
  publisher={Princeton University Press}
}

@article{Bianchi11,
  title={Overborrowing and systemic externalities in the business cycle},
  author={Bianchi, Javier},
  journal={American Economic Review},
  volume={101},
  number={7},
  pages={3400--3426},
  year={2011},
  publisher={American Economic Association}
}

@book{OR96,
  title={Foundations of international macroeconomics},
  author={Obstfeld, Maurice and Rogoff, Kenneth},
  year={1996},
  publisher={MIT press}
}
\newpage
\appendix 

\clearpage 
\newpage
\section{Proofs}

\subsection{Proof of Proposition \ref{prop1}} \label{proof:prop1}
By definition, changes in the aggregate price index satisfy
\begin{align}
    \widehat{P} &= \bar{\bm{b}}_D^T \widehat{\bm{P}}_D + \bar{\bm{b}}_M^T \widehat{\bm{P}}_M
\end{align}

By Shephard's lemma the vector of domestic prices can be written as
\begin{align*}
    \widehat{\bm{P}}_D& = - \widehat{\bm{Z}}+ \bm{A} \widehat{\bm{W}} + \bm{\Omega}\widehat{\bm{P}}_D  + \bm{\Psi}_D \bm{\Gamma} \widehat{\bm{P}}_M,
\end{align*}
where this result follows since in equilibrium $MC_i = P_i$ for all $i = 1,2,...,N$.

Inverting this system yields
\begin{align*}
    \widehat{\bm{P}}_D& = - \bm{\Psi}_D\widehat{\bm{Z}}+ \bm{\Psi}_D \bm{A} \widehat{\bm{W}} + \bm{\Psi}_D \bm{\Gamma} \widehat{\bm{P}}_M
\end{align*}
Using the definitions and equilibrium results 
\begin{align}
 \bar{\bm{b}}_D^T \widehat{\bm{P}}_D &=    \bm{\bar{b}}_D^T\left[-\bm{\Psi}_D \widehat{\bm{Z}} + \bm{\Psi}_D \bm{A} \widehat{\bm{W}} + \bm{\Psi}_D \bm{\Gamma} \widehat{\bm{P}}_M  \right]
\end{align}
Thus CPI changes can be written as 
\begin{align}
    \widehat{P}&= -\bm{\bar{b}}_D^T\bm{\Psi}_D \widehat{\bm{Z}} + \bm{\bar{b}}_D^T\bm{\Psi}_D \bm{A} \widehat{\bm{W}} + ( \bar{\bm{b}}^T_M + \bar{\bm{b}}_D^T\bm{\Psi}_D \bm{\Gamma})\widehat{\bm{P}}_M
\end{align}
To get the expression in the text, simply note that goods market clearing condition and the definition of factor shares can be written as 
\begin{align*}
    \bm{\bar{\lambda}}^T&= (\bar{\bm{b}}_D^T + \bar{\bm{x}}^T)\bm{\Psi}_D\Longrightarrow \bar{\bm{b}}_D^T\bm{\Psi}_D = \bm{\bar{\lambda}}^T - \bar{\bm{x}}^T\bm{\Psi}_D \\
    \bm{\bar{\Lambda}}^T&= \bm{\bar{\lambda}}^T\bm{A} \Longrightarrow \bar{\bm{b}}_D^T\bm{\Psi}_D\bm{A} = \bm{\bar{\Lambda}}^T -\bar{\bm{x}}^T\bm{\Psi}_D\bm{A}. 
\end{align*}
Replacing these expressions with changes in the CPI index yields the result.

\subsection{Proof of Proposition \ref{prop2}} \label{proof:prop2}

To prove this proposition, notice that Hulten's theorem in an efficient closed economy with inelastic factor supplies implies that changes in real GDP must satisfy 
\begin{align*}
    \widehat{Y}&= \bm{\lambda}^T\widehat{\bm{Z}} + \bm{\Lambda}^T\widehat{\bm{L}}.
\end{align*}
In turn, changes in factor supply have to comply $\widehat{\bm{L}}= -\widehat{\bm{W}} + \widehat{nGDP}\bm{1}_F + \widehat{\bm{\Lambda}}$, which upon replacing this expression above I get 
\begin{align*}
  \widehat{Y}  &=\bm{\lambda}^T\widehat{\bm{Z}}+ \widehat{nGDP} - \bm{\Lambda}^T\widehat{\bm{W}}.
\end{align*}
The definition of changes in nominal GDP has a quantity and a price component. The quantity component refers to $\widehat{Y}$ (real GDP). The price component is the GDP deflator, which is equivalent to CPI in a closed economy\footnote{This provided that we interpret CPI broadly to include all final uses different from intermediate inputs, such as investment, government expenditure, and so on.}. This allows us to write changes in nominal GDP as 
\begin{align*}
    \widehat{nGDP} &=  \widehat{Y} + \widehat{CPI}\Longrightarrow 
     \widehat{CPI} = \widehat{nGDP} -  \widehat{Y} 
\end{align*}
Combining the above with the changes in real GDP as a function of productivity, factor prices, and nominal GDP, I get
\begin{align*}
   \widehat{CPI} &= \widehat{nGDP} - (\bm{\lambda}^T\widehat{\bm{Z}}+ \widehat{nGDP} - \bm{\Lambda}^T\widehat{\bm{W}}) = -\bm{\lambda}^T\widehat{\bm{Z}} + \bm{\Lambda}^T\widehat{\bm{W}},
\end{align*}
which was the desired expression.

\subsection{Proof of Proposition \ref{prop3}} \label{proof:prop3}

As stated in the text
\begin{align*}
      \widehat{\bm{W}} &= \widehat{\bm{\bar{\Lambda}}} + \bm{1}_F \widehat{\mathcal{M}} - \widehat{\bm{\bar{L}}}.  
\end{align*}

Multiplying the above expression by the weight on wages in CPI from equation (\ref{eq:main}), I get 
\begin{small}
\begin{align*}
  \left(\bm{\bar{\Lambda}}^T - \bm{\bar{x}}^T \bm{\Psi}_D\bm{A}\right)\widehat{\bm{W}} & =  \left(\bm{\bar{\Lambda}}^T - \bm{\bar{x}}^T \bm{\Psi}_D\bm{A}\right)(\widehat{\bm{\bar{\Lambda}}} + \bm{1}_F \widehat{\mathcal{M}} - \widehat{\bm{\bar{L}}})
\end{align*}
\end{small}

Now note that, in general, the budget constraint of a consumer in a small open economy can be written as 
\begin{align*}
    \mathcal{M} + T &= GDP\Longrightarrow \widehat{nGDP} = \frac{\mathcal{M}}{nGDP}\widehat{\mathcal{M}}+ \frac{\mathrm{d}T}{nGDP}.
\end{align*}
By definition, 
\begin{align*}
    \Lambda_f &= \bar{\Lambda}_f \frac{\mathcal{M}}{nGDP}\Longrightarrow \widehat{\bar{\Lambda}}_f = \widehat{\Lambda}_f - \widehat{\mathcal{M}} +\widehat{nGDP},\\
    \bm{\bar{\Lambda}}^T\bm{1}_F = \sum\limits_{f\in F}\bar{\Lambda}_f &= \frac{nGDP}{\mathcal{M}} =\frac{\mathcal{M} + T}{\mathcal{M}}   \Longrightarrow  \bm{\bar{\Lambda}}^T\bm{\bar{\Lambda}} = \sum\limits_{f \in F}\mathrm{d}\bar{\Lambda}_f = \mathrm{d}\left(1 + \frac{T}{\mathcal{M}}\right) = \frac{\mathrm{d}T }{\mathcal{M}} - \frac{T}{\mathcal{M}}\widehat{\mathcal{M}}
\end{align*}

Then, 
\begin{align*}
     \left(\bm{\bar{\Lambda}}^T - \bm{\bar{x}}^T \bm{\Psi}_D\bm{A}\right)\widehat{\bm{W}} & = \bm{\bar{\Lambda}}^T\widehat{\bm{\bar{\Lambda}}}  - \bm{\bar{x}}^T \bm{\Psi}_D\bm{A}\widehat{\bm{\bar{\Lambda}}} + \left(\bm{\bar{\Lambda}}^T - \bm{\bar{x}}^T \bm{\Psi}_D\bm{A}\right)\bm{1}_F \widehat{\mathcal{M}} - \left(\bm{\bar{\Lambda}}^T - \bm{\bar{x}}^T \bm{\Psi}_D\bm{A}\right)\widehat{\bm{\bar{L}}}\\
     &= \frac{\mathrm{d}T }{\mathcal{M}} - \frac{T}{\mathcal{M}}\widehat{\mathcal{M}}  - \bm{\bar{x}}^T \bm{\Psi}_D\bm{A}\widehat{\bm{\bar{\Lambda}}} + \left(\bm{\bar{\Lambda}}^T - \bm{\bar{x}}^T \bm{\Psi}_D\bm{A}\right)\bm{1}_F \widehat{\mathcal{M}} - \left(\bm{\bar{\Lambda}}^T - \bm{\bar{x}}^T \bm{\Psi}_D\bm{A}\right)\widehat{\bm{\bar{L}}}\\
     &= \frac{\mathrm{d}T }{\mathcal{M}} - \bm{\bar{x}}^T \bm{\Psi}_D\bm{A}\widehat{\bm{\bar{\Lambda}}} + \left(\left(\bm{\bar{\Lambda}}^T - \bm{\bar{x}}^T \bm{\Psi}_D\bm{A}\right)\bm{1}_F - \frac{T}{\mathcal{M}}\right) \widehat{\mathcal{M}} - \left(\bm{\bar{\Lambda}}^T - \bm{\bar{x}}^T \bm{\Psi}_D\bm{A}\right)\widehat{\bm{\bar{L}}}\\
     &= \frac{\mathrm{d}T }{\mathcal{M}} - \bm{\bar{x}}^T \bm{\Psi}_D\bm{A}\widehat{\bm{\bar{\Lambda}}} + \left(\frac{\mathcal{M} + T}{\mathcal{M}} - \bm{\bar{x}}^T \bm{\Psi}_D\bm{A}\bm{1}_F - \frac{T}{\mathcal{M}}\right) \widehat{\mathcal{M}} - \left(\bm{\bar{\Lambda}}^T - \bm{\bar{x}}^T \bm{\Psi}_D\bm{A}\right)\widehat{\bm{\bar{L}}}\\
     &= \frac{\mathrm{d}T }{\mathcal{M}}  + \left(1- \bm{\bar{x}}^T \bm{\Psi}_D\bm{A}\bm{1}_F\right) \widehat{\mathcal{M}}- \bm{\bar{x}}^T \bm{\Psi}_D\bm{A}\widehat{\bm{\bar{\Lambda}}} - \left(\bm{\bar{\Lambda}}^T - \bm{\bar{x}}^T \bm{\Psi}_D\bm{A}\right)\widehat{\bm{\bar{L}}}
\end{align*}
Replacing this expression into equation (\ref{eq:main}), I get 
\begin{align}
    \widehat{P} &= -\left(\bm{\bar{\lambda}}^T - \tilde{\bm{\lambda}}^T\right)\widehat{\bm{Z}} - \underbrace{\tilde{\bm{\Lambda}}^T\widehat{\bm{\bar{\Lambda}}} - \left(\bar{\bm{\Lambda}}^T - \bm{\tilde{\bm{\Lambda}}}^T\right)\widehat{\bm{\bar{L}}} + \frac{\mathrm{d}T}{\mathcal{M}} + \left(1 - \tilde{\bm{\Lambda}}^T\bm{1}_F\right) \widehat{\mathcal{M}}}_{\text{Factor price changes}}\nonumber \\ 
    & +  \left( \bm{\bar{b}}_M^T + \bm{\bar{b}}_D^T \bm{\Psi}_D\bm{\Gamma}\right)\widehat{\bm{P}}_M, 
\end{align}
which is the expression in the main text.
\section{Two period model}\label{app:two_period}
In this section, I justify using the one-period model in the main text. The intuition from the main text is unchanged in this more complicated model.

Suppose there are two periods, $0$ and $1$. There is no uncertainty. The consumer has preferences over the consumption bundle in both periods according to some utility function $U(C)$. The consumer can also access an internationally traded bond that pays a nominal interest rate $i^*_t$. This nominal interest rate is \emph{exogenous} from the perspective of the small open economy. 

The consumer's budget constraint at times 0 and 1 read 
\begin{align*}
    P_0C_0 + \mathcal{E}_0 B_0 &= (1+i^*_{-1})\mathcal{E}_0B_{-1} + nGDP_0,\\
    P_1C_1 + \mathcal{E}_1 B_1 &= (1+i^*_{0})\mathcal{E}_1B_{0} + nGDP_1,
\end{align*}
where $P_t$ is the price index at time $t$, $C_t$ is consumption at time $t$, $B_t$ denotes asset holdings in foreign currency at time $t$, $\mathcal{E}_t$ is the nominal exchange rate at time $t$ defined as local currency per unit of foreign currency, and $nGDP_t$ denotes nominal GDP at time $t$.

Combining the two budget constraints gives the intertemporal budget constraint 
\begin{align*}
    P_0C_0 + \frac{P_1C_1}{\frac{\mathcal{E}_1}{\mathcal{E}_0}(1+i^*_0)}&= (1+i^*_{-1})\mathcal{E}_0 B_{-1}+nGDP_0 + \frac{nGDP_1}{\frac{\mathcal{E}_1}{\mathcal{E}_0}(1+i^*_0)}
\end{align*}
Under perfect mobility of capital flows, we have the \emph{no-arbitrage condition}\footnote{Under uncertainty, this is simply the uncovered interest parity condition (UIP). }
\begin{align}
    (1+i_0) &= (1+i^*_0)\frac{\mathcal{E}_1}{\mathcal{E}_0}. \label{eq:no_arbitrage}
\end{align}
Where given perfect foresight, there is no expectation regarding the future level of the exchange rate.
I can also get this condition by adding a domestic bond in zero net supply. This does not change any of the conclusions below.

The no-arbitrage condition is important for the small open economy as it clearly illustrates the fact that the Central bank has two instruments to set the nominal interest rate $i_0$: it can either choose $i_0$ directly and let the exchange rate $\mathcal{E}_0$ adjust. Or it can pick $\mathcal{E}_0$ and let the nominal interest rate accommodate to comply with that rule.

We can thus rewrite the maximization problem as solving the following program 
\begin{footnotesize}
\begin{align}
    \max_{C_0,C_1}\qquad U(C_0) + \beta U(C_1)\qquad \text{s.t}\quad  P_0C_0 + \frac{P_1C_1}{(1+i_0)}&= (1+i^*_{-1})\mathcal{E}_0 B_{-1}+nGDP_0 + \frac{nGDP_1}{(1+i_0)}
\end{align}
\end{footnotesize}
Letting $\lambda$ be the multiplier on the intertemporal budget constraint, I have 
\begin{align*}
     U'(C_0) &= \lambda P_0\\
    \beta U'(C_1) &= \lambda \frac{P_1}{(1+i_0)}
\end{align*}
Combining both equations 
\begin{align*}
\frac{U'(C_0)}{P_0} &= \beta \frac{U'(C_1)}{P_1}(1+i_0)
\end{align*}
Assume $U(C) = \log C$, as in \cite{GL07} and \cite{BF22}, so that the intertemporal elasticity of substitution is unitary.  Then 
\begin{align*}
    \beta P_0C_0 (1+i_0)&= P_1C_1
\end{align*}
As argued in \cite{BF22}, who in turn relied on an argument made in \cite{Krugman98} and \cite{EK12}, we can understand this model as assuming that anything happening at $t = 1$ is labeled as ``future". The assumption here is isomorphic to an infinite horizon model where an unexpected shock happens at $t =0$, and the economy returns to the long-run equilibrium from $t=1$ onwards. For all practical purposes, this means I assume that any variables at $t=1$ are exogenously given. Then, from the Euler equation, I have 
\begin{align*}
    P_0C_0 &= \frac{P_1C_1}{\beta(1+i_0)} = \frac{P_1C_1}{\beta (1+i^*_0)}\frac{\mathcal{E}_0}{\mathcal{E}_1}
\end{align*}
Therefore, since ($\beta, i^*_0, P_1C_1,\mathcal{E}_1$) are exogenous, the nominal exchange rate $\mathcal{E}_0$ is determined via the no-arbitrage condition, equation (\ref{eq:no_arbitrage}). This equation provides a value for current expenditure in local currency, $P_0C_0$. 

For simplicity, suppose $B_{-1} = 0$. Replacing the Euler equation in the intertemporal budget constraint. 
\begin{align*}
    P_0C_0 + \frac{P_1C_1}{(1+i_0)}&= nGDP_0 + \frac{nGDP_1}{(1+i_0)}\\
    P_0C_0 + \frac{\beta P_0C_0 (1+i_0)}{(1+i_0)}&= nGDP_0 + \frac{nGDP_1}{(1+i_0)}\\
    P_0 C_0 &= \frac{1}{(1+\beta)}\left(nGDP_0 + \frac{nGDP_1}{(1+i_0)} \right)\\
     P_0 C_0 &= \frac{1}{(1+\beta)}\left(nGDP_0 + \frac{nGDP_1}{(1+i^*_0)}\frac{\mathcal{E}_0}{\mathcal{E}_1} \right).
\end{align*}
Note that given $(\mathcal{E}_0, P_0C_0, nGDP_1, i_0)$ from the Euler equation and no-arbitrage condition, the latter equation pins down $nGDP_0$.
\subsection{Solving for consumption at time 0, $C_0$.} Before solving for consumption, let me introduce the real exchange rate, $\mathcal{Q}_0$ as 

\begin{align*}
    \mathcal{Q}_0 &= \frac{\mathcal{E}_0P^F_0}{P_0},
\end{align*}
where $P^F_0$ is the rest of the world price index, which is exogenous from the perspective of the small open economy. Following, \cite{SGUW22}, I can write this foreign price index as $P^F_0 = \mathcal{P}^F(\lbrace P^*_{m,0}\rbrace_{m \in M}, \lbrace P^*_{j,0}\rbrace_{j \in N^*})$, where $M$ is the same set of goods that are imported by the small open economy and $N^*$ is the set of all other goods consumed abroad by the foreign economy. Since I assume that all prices $P^*_k$ for $k \in M \cup N^*$ are exogenous from the perspective of the small open economy, this allows me to write 
\begin{align*}
    \frac{P^F_0}{P^*_{m_0,0}} = \mathcal{P}^F(1, \lbrace P^*_{m,0}/P^*_{0,0}\rbrace_{m \in M\ 0}, \lbrace P^*_{j,0}/P^*_{0,0}\rbrace_{j \in N^*})
\end{align*}
Using this in the definition of the real exchange rate and the law of one price for imported goods $P_{m,0} = \mathcal{E}_0 P^*_{m,0}$ for all $m \in M$, then 
\begin{align*}
    \mathcal{Q}_0 &= \frac{P^F_0/P^*_{m_0,0}}{P_0/P_{m_0,0}}
\end{align*}
Since the numerator is exogenous, the real exchange rate can be written as the relative price of $CPI$ to one of the imported goods $P_{m_0,0}$. Let set $P^F_0/P^*_{m_0,0} = 1$, then 
\begin{align*}
    \mathcal{Q}_0 &= \frac{P_{m_0,0}}{P_0}
\end{align*}
Note that consumption at time 0 is only a function of this real exchange rate since from the Euler equation 
\begin{align*}
    C_0 &=  \frac{P_1C_1}{\beta \mathcal{E}_1(1+i^*_0)}\frac{\mathcal{E}_0}{P_0} = \frac{P_1C_1}{\beta \mathcal{E}_1(1+i^*_0)}\mathcal{Q}_0 = \frac{E_1}{\beta \mathcal{E}_1 (1+i^*_0)} \mathcal{Q}_0 
\end{align*}
Changes in the real exchange rate represent changes in all prices relative to the imported good $m_0$.

\begin{align*}
    \widehat{P}_0 - \widehat{P}_{m_0,0} &= \sum\limits_{i \in N}\frac{P_iC_i}{PC} (\widehat{P}_i-\widehat{P}_{m_0,0}) + \sum\limits_{m \in M}\frac{P_mC_m}{PC} (\widehat{P}_m-\widehat{P}_{m_0,0}) = \bm{\bar{b}}_D^T (\bm{\widehat{P}}_D-\bm{1}_N\widehat{P}_{m_0,0}) + \bm{\bar{b}}_M^T (\widehat{\bm{P}}^*_M - \bm{1}_M \widehat{P}^*_{m_0,0})\\
    \widehat{P}_0 - \widehat{P}_{m_0,0}&= \bm{\bar{b}}_D^T (-\bm{\Psi}\widehat{\bm{Z}} + \bm{\Psi}\bm{A}(\widehat{\bm{W}}-\bm{1}_F\widehat{P}_{m_0,0}) + \bm{\Psi}\bm{\Gamma}(\widehat{\bm{P}}^*_M-\bm{1}_M\widehat{P}^*_{m_0,0})) + \bm{\bar{b}}_M^T (\widehat{\bm{P}}^*_M-\bm{1}_M\widehat{P}^*_{m_0,0}) \\
    \widehat{P}_0 - \widehat{P}_{m_0,0}&= -\bm{\bar{b}}_D^T\bm{\Psi}\widehat{\bm{Z}} + \bm{\bar{b}}_D^T\bm{\Psi}\bm{A}(\widehat{\bm{W}}-\bm{1}_F\widehat{P}_{m_0,0}) + (\bm{\bar{b}}_D^T\bm{\Psi}\bm{\Gamma} +\bm{\bar{b}}_M^T)(\widehat{\bm{P}}^*_M-\bm{1}_M\widehat{P}^*_{m_0,0})\\
    \widehat{Q}_0 &= -(\widehat{P}_0 - \widehat{P}_{m_0,0}) = \bm{\bar{b}}_D^T\bm{\Psi}\widehat{\bm{Z}} - \bm{\bar{b}}_D^T\bm{\Psi}\bm{A}(\widehat{\bm{W}}-\bm{1}_F\widehat{P}_{m_0,0}) - (\bm{\bar{b}}_D^T\bm{\Psi}\bm{\Gamma} +\bm{\bar{b}}_M^T)(\widehat{\bm{P}}^*_M-\bm{1}_M\widehat{P}^*_{m_0,0})
\end{align*}
It follows that consumption changes satisfy 
\begin{align*}
    \widehat{C}_0 &= \widehat{E}_1 - \widehat{\beta} - \widehat{\mathcal{E}}_1 - \widehat{(1+i^*_0)}    + \widehat{Q}_0\\
    &= \widehat{E}_1 - \widehat{\beta} - \widehat{\mathcal{E}}_1 - \widehat{(1+i^*_0)} + \bm{\bar{b}}_D^T\bm{\Psi}\widehat{\bm{Z}} - \bm{\bar{b}}_D^T\bm{\Psi}\bm{A}(\widehat{\bm{W}}-\bm{1}_F\widehat{P}_{m_0,0}) - (\bm{\bar{b}}_D^T\bm{\Psi}\bm{\Gamma} +\bm{\bar{b}}_M^T)(\widehat{\bm{P}}^*_M-\bm{1}_M\widehat{P}^*_{m_0,0}),
\end{align*}
this illustrates how consumption at time $0$ is not pinned down by real GDP, $Y_0$, as is the case in the closed economy model. Rather, it is pinned down by the real exchange rate, $\mathcal{Q}_0$, which in turn depends on factor prices in units of good $m_0$.

Using the fact that $P_0C_0 = E_0$ is given once I set either $i_0$ or $\mathcal{E}_0$, then changes in the price index satisfy 
\begin{align*}
    \widehat{P}_0 &= \widehat{E}_0 - \widehat{C}_0\\
    &= \widehat{E}_0 - \widehat{E}_1 + \widehat{\beta} + \widehat{\mathcal{E}}_1 + \widehat{(1+i^*_0)} - \bm{\bar{b}}_D^T\bm{\Psi}\widehat{\bm{Z}} + \bm{\bar{b}}_D^T\bm{\Psi}\bm{A}(\widehat{\bm{W}}-\bm{1}_F\widehat{P}_{m_0,0}) + (\bm{\bar{b}}_D^T\bm{\Psi}\bm{\Gamma} +\bm{\bar{b}}_M^T)(\widehat{\bm{P}}^*_M-\bm{1}_M\widehat{P}^*_{m_0,0})\\
     &= \widehat{\mathcal{E}}_1 + \widehat{(1+i^*_0)} - \widehat{(1+i_0)}   - \bm{\bar{b}}_D^T\bm{\Psi}\widehat{\bm{Z}} + \bm{\bar{b}}_D^T\bm{\Psi}\bm{A}(\widehat{\bm{W}}-\bm{1}_F\widehat{P}_{m_0,0}) + (\bm{\bar{b}}_D^T\bm{\Psi}\bm{\Gamma} +\bm{\bar{b}}_M^T)(\widehat{\bm{P}}^*_M-\bm{1}_M\widehat{P}^*_{m_0,0})
\end{align*}

\paragraph{From real wages to aggregate demand and inelastic labor supply changes.}
To solve the model in terms of factor quantities, let me define real wages, in terms of imported good $m_0$, as a function of these, factor shares, and the import price (denominated in foreign currency)
\begin{align*}
    \widehat{W}_{f,0} - \widehat{P}_{m_0,0} &= \widehat{\bar{\Lambda}}_f -\widehat{\bar{L}}_{f,0} + \widehat{E}_0 -\widehat{P}_{m_0,0}\\
    &= \widehat{\bar{\Lambda}}_f -\widehat{\bar{L}}_{f,0} + (\widehat{E}_0 -\widehat{\mathcal{E}}_0) - \widehat{P}^*_{m_0,0}\\
    \widehat{\bm{W}} - \bm{1}_F \widehat{P}_{m_0,0}&= \widehat{\bm{\bar{\Lambda}}} - \widehat{\bar{\bm{L}}} + \bm{1}_F((\widehat{E}_0 -\widehat{\mathcal{E}}_0) - \widehat{P}^*_{m_0,0}),
\end{align*}
where the second line follows from the law of one price. 

Note that expenditure denominated in foreign currency is exogenous from the Euler equation
\begin{align*}
    \frac{E_0}{\mathcal{E}_0} &= \frac{E_1}{\mathcal{E}_1}\frac{1}{\beta (1+i^*_0)}.
\end{align*}

Thus $(\widehat{E}_0 -\widehat{\mathcal{E}}_0)$ represents an \emph{aggregate demand shifter}.
It increases if the consumer becomes more impatient ($\beta$ declines), future expenditure in local currency increases ($E_1 = P_1C_1)$, the interest rate in foreign currency declines, $(1+i^*_0)$, or the exchange rate in the future goes down, $\mathcal{E}_1$.  

Set $\widehat{P}^*_{m0,0} = 0$ to simplify the exposition. Combining expenditure at time $0$ in foreign currency with the expression for CPI changes, I get 
\begin{align*}
    \widehat{P}_0 &= \widehat{\mathcal{E}}_0 - \bm{\bar{b}}_D^T\bm{\Psi}\widehat{\bm{Z}} + \bm{\bar{b}}_D^T\bm{\Psi}\bm{A}\left(\widehat{\bm{\bar{\Lambda}}} - \widehat{\bar{\bm{L}}} + \bm{1}_F(\widehat{E}_0 -\widehat{\mathcal{E}}_0)\right) + (\bm{\bar{b}}_D^T\bm{\Psi}\bm{\Gamma} +\bm{\bar{b}}_M^T)\widehat{\bm{P}}^*_M,
\end{align*}
which can be written as 
\begin{align}
    \widehat{P}_0  &= \underbrace{\widehat{\mathcal{E}}_0}_{\text{Nominal Anchor}} + \underbrace{\bm{\bar{b}}_D^T\bm{\Psi}\bm{A}\bm{1}_F(\widehat{E}_0 -\widehat{\mathcal{E}}_0)}_{\text{Aggregate demand shifter}} - \underbrace{\bm{\bar{b}}_D^T\bm{\Psi}\widehat{\bm{Z}}}_{\text{Technology effects}} + \underbrace{\bm{\bar{b}}_D^T\bm{\Psi}\bm{A}\widehat{\bm{\bar{\Lambda}}}}_{\text{Factor share reallocation}} \nonumber\\
    &-\underbrace{\bm{\bar{b}}_D^T\bm{\Psi}\bm{A} \widehat{\bar{\bm{L}}}}_{\text{Factor supplies}}  + \underbrace{(\bm{\bar{b}}_D^T\bm{\Psi}\bm{\Gamma} +\bm{\bar{b}}_M^T)\widehat{\bm{P}}^*_M}_{\text{Import price channel}},
\end{align}
where I used the fact the no-arbitrage condition, implies that changes in the nominal exchange rate at time $0$ can be written as 
\begin{align*}
    \widehat{(1+i_0)}&= \widehat{(1+i^*_0)} + \widehat{\mathcal{E}}_1 - \widehat{\mathcal{E}}_0\Longrightarrow \widehat{\mathcal{E}}_0 = \widehat{(1+i^*_0)} + \widehat{\mathcal{E}}_1 - \widehat{(1+i_0)}.
\end{align*}
This provides a nominal anchor at time $0$.

\subsection{Mapping the net transfer, $T$.}
We can use the two-period model above to justify the exogenous net transfer in the static setup $T$. To see this, write
\begin{align*}
    T &= nGDP_0 - P_0C_0 = (1+\beta)P_0C_0 - \frac{nGDP_1}{(1+i^*_0)}\frac{\mathcal{E}_0}{\mathcal{E}_1} - P_0C_0\\
    T &= \beta P_0C_0 - \frac{nGDP_1}{(1+i^*_0)}\frac{\mathcal{E}_0}{\mathcal{E}_1} = \beta P_0C_0 - \beta \frac{nGDP_1}{P_1C_1}P_0C_0  = \beta P_0C_0\left(1 - \frac{nGDP_1}{P_1C_1}\right),
\end{align*}
this net transfer is positive or negative depending on whether nominal GDP in the future is higher or lower than future expenditure. This ultimately hinges on the difference between future consumption and income since if $nGDP_1/P_1C_1>1$; this ratio is negative, meaning $T<0$. A negative net transfer means the economy receives resources at time 0 that do not come from their own production at time $0$. In an intertemporal model, this comes from future resources. In a static model, this should come from the rest of the world. The converse also holds. Of course, if the steady state features no assets holding, this equation collapses to $T = 0$. To see why, note the budget constraint satisfies:
\begin{align*}
PC\left(1 + \frac{1}{(1+i)}\right)&= \mathcal{E}i^* B + nGDP\left(1 + \frac{1}{(1+i)}\right),
\end{align*}
If $B = 0$, then 
\begin{align*}
    PC = nGDP,
\end{align*}
and therefore $T = 0$.

\section{A Small Open Economy Dynamic Model with Production Networks}
\label{appendix:dynamic_model}
In this appendix, I briefly explore how the intuition of the model in the main text extends to a small open economy \emph{dynamic} environment. 
\paragraph{Environment.} The model is a variant of the canonical importable, exportable, and non-tradable model (MXN) as in Chapter 8 of \cite{SGU17Book}, where I remove capital from the model. 

Time is discrete, indexed by $t$, and runs forever. Households consume exportables, importables, and non-tradables. Non-tradable and exportable goods are produced using labor and intermediate inputs from other sectors. This intermediate input-output structure is one of the novelties of the model. 

Financial markets are incomplete. The domestic household has access to two bonds. The first is a domestic bond denominated in local currency and pays a nominal interest rate $i_t$. The second is a foreign bond denominated in foreign currency and pays a nominal interest rate $i^*_t$. The small open economy takes this latter foreign interest rate as given. 
\paragraph{Household.} The household owns labor and consumes the three goods. Labor is supplied inelastically. The consumption aggregator is of the CES form 
\begin{align}
    C &= \left(b_N^{\frac{1}{\chi}} C_N^{\frac{\chi - 1}{\chi}} +b_M^{\frac{1}{\chi}} C_M^{\frac{\chi - 1}{\chi}}  + (1-b_N-b_M)^{\frac{1}{\chi}} C_X^{\frac{\chi - 1}{\chi}}\right)^{\frac{\chi}{\chi - 1}},
\end{align}
where $C_N$ is consumption of non-tradables, $C_M$ is consumption of importables, and $C_X$ is consumption of the exportable good. $(b_N, b_M)$ are the expenditure shares on non-tradable and importables at the symmetric price steady-state. In turn, $1-b_N-b_M$ is the expenditure share on exportable goods. Finally, $\chi$ is the elasticity of substitution across the different goods.

I solve the household's problem in two steps. In the first step, I solve for the dynamic path of $\lbrace C_t, B_t, B^*_t\rbrace_{t=0}^\infty$. Conditional on knowing the path of $C_t$, we can solve for $(C_{Nt}, C_{Mt}, C_{Xt})$ at each instant $t$. This way of solving the problem allows me to simplify the exposition and is the route followed in, for example, Chapter 4 of \cite{OR96}.\footnote{Specifically, see Section 4.4.1. of \cite{OR96}.}

The dynamic problem of the household is as follows. Taking as given paths of prices $\lbrace P_t, W_t, \mathcal{E}_t\rbrace $, interest rates $\lbrace i_t, i^*_t\rbrace$, and labor endowment $\lbrace \bar{L}_t\rbrace $, the consumer solves the following program 
\begin{align*}
    \max_{\lbrace C_t, B_t, B^*_t \rbrace_{t=0}^\infty}&\mathbb{E}_0\sum\limits_{t=0}^\infty \beta^t \frac{C_t^{1-\sigma}-1}{1-\sigma}  \\
    \text{subject to } P_tC_t + \mathcal{E}_t B^*_t + B_t &\leq W_t\bar{L}_t + (1+i^*_{t-1})\mathcal{E}_tB^*_{t-1} +  (1+i_{t-1})B_{t-1},
\end{align*}
where $P_t$ is the consumer's price index, $B^*_t$ is assets holdings of a foreign bond, $B_t$ is assets holding of a domestic bond, $\mathcal{E}_t$ is the nominal exchange rate defined as units of home currency per unit of foreign currency, $W_t$ is the wage rate, $(1+i_t)$ is the gross interest rate in the domestic bond and $(1+i^*_t)$ is the gross interest rate in the foreign bond. $\beta \in (0,1)$ is a discount factor. 

Letting $\beta^t\mu_t$ be the Lagrange multiplier on the flow budget constraint, we get the following first-order conditions
\begin{align}
    C_t: C_t^{-\sigma}&= \lambda_tP_t,\\
    B_t: \lambda_t &= \beta (1+i_t)\mathbb{E}_t\lambda_{t+1},\\
    B^*_t: \lambda_t &= \beta (1+i^*_t)\mathbb{E}_t\lambda_{t+1} \frac{\mathcal{E}_{t+1}}{\mathcal{E}_t},
\end{align}
plus the budget constraint.

Conditional on $C$, the intratemporal problem solves the following program (ignoring time subscripts)
\begin{align}
    \min_{C_N, C_M, C_X} P_NC_N + P_MC_M + P_XC_X \text{ subject to } C \geq \bar{C}.
\end{align}
That is, taking as given prices $(P_N, P_M, P_X)$, a consumption aggregator function $C$ and a level of aggregate consumption $\bar{C}$; the household minimizes its expenditure to achieve $\bar{C}$.

The conditional demands that solve this problem are 
\begin{align}
    C_N&= b_N \left( \frac{P_N}{P}\right)^{-\chi}C;\quad C_M= b_M \left( \frac{P_M}{P}\right)^{-\chi}C;\quad C_X= (1-b_N-b_M) \left( \frac{P_X}{P}\right)^{-\chi}C,
\end{align}
where the price index $P$ satisfies 
\begin{align}
    P &= (b_NP_N^{1-\chi} + b_MP_M^{1-\chi} + (1-b_N-b_M)P_X^{1-\chi})^{\frac{1}{1-\chi}}
\end{align}
CPI Inflation is thus defined as 
\begin{align}
    \pi_t &= \log P_t - \log P_{t-1}
\end{align}
\paragraph{Production Side.} There are two producing sectors: non-tradable ($N$) and exportable ($X$). I omit time indices whenever it causes no confusion.

Gross output of sector $i \in \lbrace N, X \rbrace$ is of the CES form 
\begin{align}
    Q_i &= Z_i \left(a_i^{\frac{1}{\sigma_i}} L_i^{\frac{\sigma_i-1}{\sigma_i}} + (1-a_i)^{\frac{1}{\sigma_i}} M_i^{\frac{\sigma_i-1}{\sigma_i}} \right)^{\frac{\sigma_i}{\sigma_i-1}},
\end{align}
where $L_i$ is labor demand and $M_i$ is an intermediate input bundle. $\sigma_i$ is the elasticity of substitution between labor and intermediate inputs. $Z_i$ is a sector-specific productivity level. These productivity levels are exogenous. Finally, $a_i$ represents the labor share in total costs (sales) at the symmetric price equilibrium.

The intermediate input bundle ($M_i$) aggregates non-tradable ($M_{iN}$) and tradable ($M_{iT}$) intermediate inputs according to another CES layer 
\begin{align}
    M_{i}&= \left(\omega_i^{\frac{1}{\varepsilon_i}} M_{iN}^{\frac{\varepsilon_i-1}{\varepsilon_i}} + (1-\omega_i)^{\frac{1}{\varepsilon_i}} M_{iT}^{\frac{\varepsilon_i-1}{\varepsilon_i}}\right)^{\frac{\varepsilon_i}{\varepsilon_i - 1}},
\end{align}
where $\varepsilon_i$ is the elasticity of substitution between non-tradable and tradable intermediate inputs. $\omega_i$ represents the expenditure share on non-tradable intermediates out of total intermediate input spending. Therefore, $1-\omega_i$ is the expenditure share on tradable goods.

The tradable intermediate input bundle combines the exportable good and the imported input 
\begin{align}
    M_{iT}&= \left(\omega_{iX}^{\frac{1}{\varepsilon^T_i}}M_{iX}^{\frac{\varepsilon^T_i-1}{\varepsilon^T_i}} + (1-\omega_{iX})^{\frac{1}{\varepsilon^T_i}}M_{iM}^{\frac{\varepsilon^T_i-1}{\varepsilon^T_i}}\right)^{\frac{\varepsilon^T_i}{\varepsilon^T_i - 1}},
\end{align}
where $\varepsilon^T_i$ is the elasticity of substitution between importable and exportable goods. $\omega_{iX}$ is the expenditure share on exportable goods as a share of intermediate spending on tradable goods (both exportable and importable).

Cost minimization at each CES layer delivers the following conditional demands 
\begin{align}
    L_i &= a_i \left(\frac{W}{MC_i} \right)^{-\sigma}Z_i^{\sigma_i - 1}Q_i; \quad M_i = (1-a_i \left(\frac{P^I_i}{MC_i} \right)^{-\sigma}Z_i^{\sigma_i - 1}Q_i\\
    M_{iN}&= \omega_i \left(\frac{P_N}{P^I_i}\right)^{-\varepsilon_i}M_i; \quad M_{iT}= (1-\omega_i) \left(\frac{P^T_i}{P^I_i}\right)^{-\varepsilon_i}M_i\\
    M_{iX}&= \omega_{iX}\left(\frac{P_X}{P^T_i}\right)^{-\varepsilon^T_i}M_{iT};\quad M_{iM}= (1-\omega_{iX})\left(\frac{P_M}{P^T_i}\right)^{-\varepsilon^T_i}M_{iT}
\end{align}
with marginal costs and price indices 
\begin{align}
    MC_i&= Z_i^{-1}\left(a_i W^{1-\sigma_i} + (1-a_i)(P^I_i)^{1-\sigma_i} \right)^{\frac{1}{1-\sigma_i}},\\
    P^I_i&= \left(\omega_i P_N^{1-\varepsilon_i} + (1-\omega_i)(P^T_i)^{1-\varepsilon_i} \right)^{\frac{1}{1-\varepsilon_i}},\\
    P^T_i&= \left(\omega_{iX} P_X^{1-\varepsilon^T_i} + (1-\omega_{iX})(P_M)^{1-\varepsilon^T_i} \right)^{\frac{1}{1-\varepsilon^T_i}},
\end{align}
where $MC_i$ stands for marginal cost, $P^I_i$ is the price index of the intermediate input bundle, and $P^T_i$ is the price index of the tradable intermediate input bundle.

\paragraph{Law of one price.} I assume the law of one price holds for the exportable and importable good. This means 
\begin{align}
    P_{Xt}&= P^*_{Xt}\mathcal{E}_t\; \quad P_{Mt}= P^*_{Mt}\mathcal{E}_t
\end{align}

\paragraph{Nominal Anchor.} Households require local currency to buy the consumption bundles. I introduce this notion as a cash-in-advance constraint. I impose this as an additional aggregate equation 
\begin{align}
    \mathcal{M}_t &= P_tC_t,
\end{align}
where $\mathcal{M}_t$ is an \emph{exogenous} money supply. Since prices are fully flexible and there are no market distortions, imposing this constraint does not affect relative prices and optimal allocations. It allows me to pin down the price level and, as a by-product inflation. This is akin to a nominal GDP targeting rule, which in the open economy is an \emph{expenditure} targeting rule instead.\footnote{I can also specify a model with money in the utility function. As long as money and aggregate consumption preferences are separable,  none of the results I present here change when using such a specification. This illustration closely follows the main text.}

\paragraph{Exogenous processes.} In the model there are 6 exogenous processes $(Z_{Nt}, Z_{Xt}, P^*_{Xt}, P^*_{Mt}, \mathcal{M}_t, \bar{L}_t)$. For the purposes of the exercise, I explore changes in $(Z_{Nt}, P^*_{Mt})$. I assume all processes follow AR(1) in logs. That is 
\begin{align}
    \log Z_{Nt}&= \rho_{Z_N}\log Z_{Nt-1} + \nu^N_{t}\\
    \log P^*_{Mt} &= \rho_{P_M}\log P^*_{Mt-1} + \nu^{P_M}_t,
\end{align}
where $(\nu^{N}_t, \nu^{P_M}_t)$ are disturbances. 
\paragraph{Equilibrium.} Since the domestic bond is traded only within the country, we have that $B_t$ = 0. The non-tradable market clearing condition and the labor market clearing condition is 
\begin{align}
    Q_{Nt}&= C_{Nt} + M_{NNt} + M_{XNt}\\
    \bar{L}_t &= L_{Nt} + L_{Xt}
\end{align}
Combining these three conditions into the consumer's budget constraint, we can write the law of motion for foreign assets as 
\begin{footnotesize}
\begin{align}
    B^*_t&= (1+i^*_{t-1})B^*_{t-1} - \frac{1}{\mathcal{E}_t}\left(P_{Xt}(C_{Xt} + M_{XXt} + M_{NXt} - Q_{Xt}) + P_{Mt}(C_{Mt} + M_{NMt} + M_{XMt}) \right)
\end{align}
\end{footnotesize}

I induce stationarity in this foreign asset position using a debt-elastic interest rate device, as in \cite{SGU03}. This means 
\begin{align}
    i^*_t &= \bar{i}^* +\psi (e^{\bar{B}^* - B^*_t} - 1),
\end{align}
where $\bar{i}^*$ and $\bar{B}^*$ are steady-state values of the interest rate and foreign assets.

\subsection{Calibration and Scenarios}
A period is a year. To assess the role of the production network structure, I consider two different scenarios that vary the exposure and dependence of the exportable and non-tradable sectors to each other via input-output linkages. 

In particular, I consider the following specifications:
\begin{enumerate}
    \item Island: ($\omega_N, \omega_X, \omega_{NX}, \omega_{XX}$) = $(1,0, 0, 0.5)$.
    \item Intersectoral linkages: ($\omega_N, \omega_X, \omega_{NX}, \omega_{XX}$) = $(0, 1, 1, 0.5)$.
\end{enumerate}

The first scenario ignores intersectoral linkages and treats both sectors as islands isolated from each other in the input-output structure. The non-tradable sector only buys intermediate input from the same sector, while the exportable sector buys from itself and the imported good. The idea of this scenario is to shut down intersectoral linkages. The second scenario assumes the non-tradable sector only uses tradable intermediate inputs, while the exportable sector only uses non-tradable intermediate inputs. Importantly, these two scenarios play around with the intermediate expenditure share distribution while keeping the total intermediate expenditure share constant at $1-a_i = 0.33$ in both sectors.

Table \ref{tab:calibration} shows the calibrated shares and parameters kept fixed in both scenarios that I borrowed from the literature. Since the consumption share on tradables (exportables plus importables) is 0.3, I use the estimate expenditure share on importables (relative to tradable expenditure) from Table 8.2 in \cite{SGU17Book}, which is $\chi_m = 0.898$ (in their notation). This implies that the consumption expenditure on importables as a share of total expenditure equals $b_M =  0.3 \times 0.898 = 0.27$. I set all elasticities in production and consumption to be Cobb-Douglas. I do this to highlight the first-order mechanisms, as it is well known that under low elasticities of substitution, the production network amplifies negative shocks on quantities and mitigates positive shocks \cite{BF19}, and thus production networks matter beyond sales shares to a second-order. All remaining parameters are standard in the small open economy literature.
\begin{table}[htbp!]
    \centering
        \caption{Calibration}
    \label{tab:calibration}
    \resizebox{\textwidth}{!}{
    \begin{tabular}{lcll}
    \toprule
    Parameters & Value & Description & Source\\
    \midrule
     \emph{Shares}\\
         $a_N = a_X$& 0.66 & Labor Share & \cite{BCORY13} \\
         $b_N$ & 0.70 & Consumption share on non-tradables & \cite{Bianchi11}\\
         $b_X$ & 0.03 & Consumption share on exportable & Table 8.2 \citep{SGU17Book} \\
         $b_M$ & 0.27 & Consumption share on importable & Table 8.2 \citep{SGU17Book}\\
         \\
         \emph{Elasticities}\\
          $\chi$ & 1 & Elasticity of substitution in consumption & Cobb-Douglas specification\\
         $\sigma$ & 2 & Intertemporal Elasticity of Substitution& Table 8.2 \citep{SGU17Book}\\
         $\sigma_N = \sigma_X$ & 1 & Elasticity between value-added and intermediates & Cobb-Douglas specification\\
         $\varepsilon_N = \varepsilon_X$ & 1 & Elasticity across intermediates & Cobb-Douglas specification\\
         $\varepsilon^T_N  = \varepsilon^T_X$ & 1 & Elasticity across tradable intermediates & Cobb-Douglas specification\\
         \\
         \emph{Other Parameters}\\
          $\rho_{Z_N}=\rho_{P_M}$ & 0.53 & AR(1) non-tradable productivity and import price & Table 7.1 \citep{SGU17Book} \\
         $\bar{B}^*$ & 0 & Steady-state foreign assets position & Zero trade balance\\
          $\psi$ & 0.000742 & Interest rate sensitivity to foreign assets & \cite{SGU03}\\
          $\bar{i}^*$& 0.04 & Steady-state foreign interest rate & \cite{Bianchi11}\\
          $\beta$ & $\frac{1}{(1+i^*)}=0.9615$ & Discount Factor\\
         \bottomrule
    \end{tabular}}
\end{table}

\subsection{Results}
I explore how inflation reacts to a one percent shock negative productivity in the non-tradable sector ($\nu^N_{t}$) and a positive import price shock ($\nu^M_{t}$). I solve the model using a first-order approximation around the non-stochastic steady state to be comparable with the model in the main text. Panel (a) of Figure \ref{fig:inflation_irf} shows the response of inflation to a negative productivity shock in the non-tradable sector, while panel (b) of Figure \ref{fig:inflation_irf} does the same for a positive import price shock. The solid purple line shows the implied impulse-response function for inflation in a model with intersectoral linkages. In contrast, the pink dashed line shows the impulse-response function for the island model.

The dynamic model confirms the intuition of the static model. Panel (a) shows that a decrease in productivity generates inflation, which is lower in an economy with intersectoral linkages (solid line) due to indirect trade. A positive import price shock, on the other hand, increases inflation, more so in a model with intersectoral linkages due to indirect trade. This highlights that the production network matters to a first-order for the productivity pass-through to inflation. Importantly, Domar weights are no longer sufficient statistics for the productivity pass-through to inflation, as they are equal in both scenarios by construction. 

The persistence in the dynamic responses is given by shocks' persistence as the model has no inherent dynamics from the production network structure. Of course, the model does have dynamics from the foreign asset position and thus is richer relative to the static model in the main section. Moreover, a one-time shock at time 0 generates inflation in period 0 but requires deflation thereafter for the price level to recover its steady-state level in the long run.

Finally, inflation volatility due to productivity shocks in the non-tradable sector is lower under the intersectoral model than the island model. The opposite is true when we consider import price shocks. These inflation impulse responses are consistent with the results in Section \ref{sec:application}, where the small open economy model with production networks decreases inflation volatility in Chile and increases it for the UK, as the production network adjustments put more weight on import prices and less on productivity shocks.
\begin{figure}[htbp!]
    \centering
        \caption{Inflation Impulse Responses}
            \label{fig:inflation_irf}
    \begin{minipage}{0.49\textwidth}
    \caption*{\scriptsize (a) Negative Productivity Shock in Non-Tradable Sector}
\includegraphics[width =\textwidth]{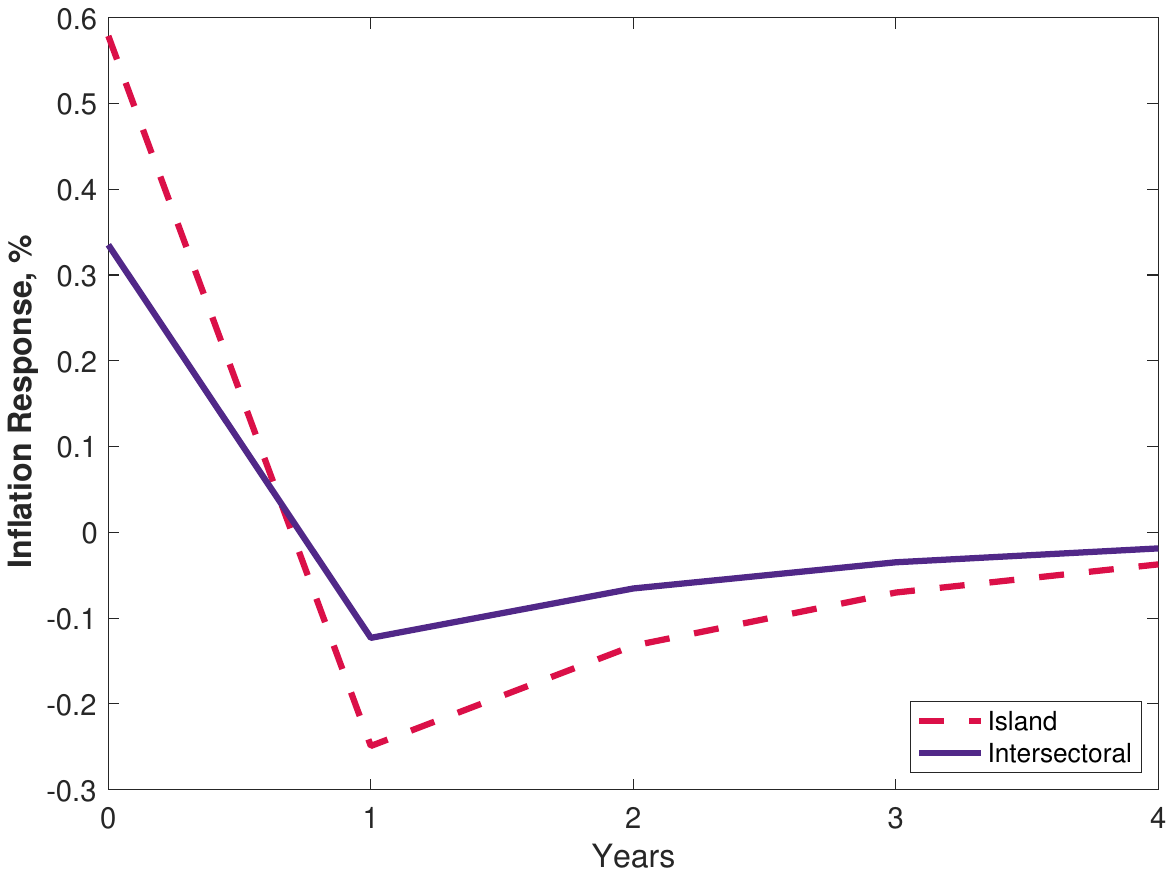}
    \end{minipage}
    \begin{minipage}{0.49\textwidth}
    \caption*{\scriptsize (b) Import Price Shock}
\includegraphics[width =\textwidth]{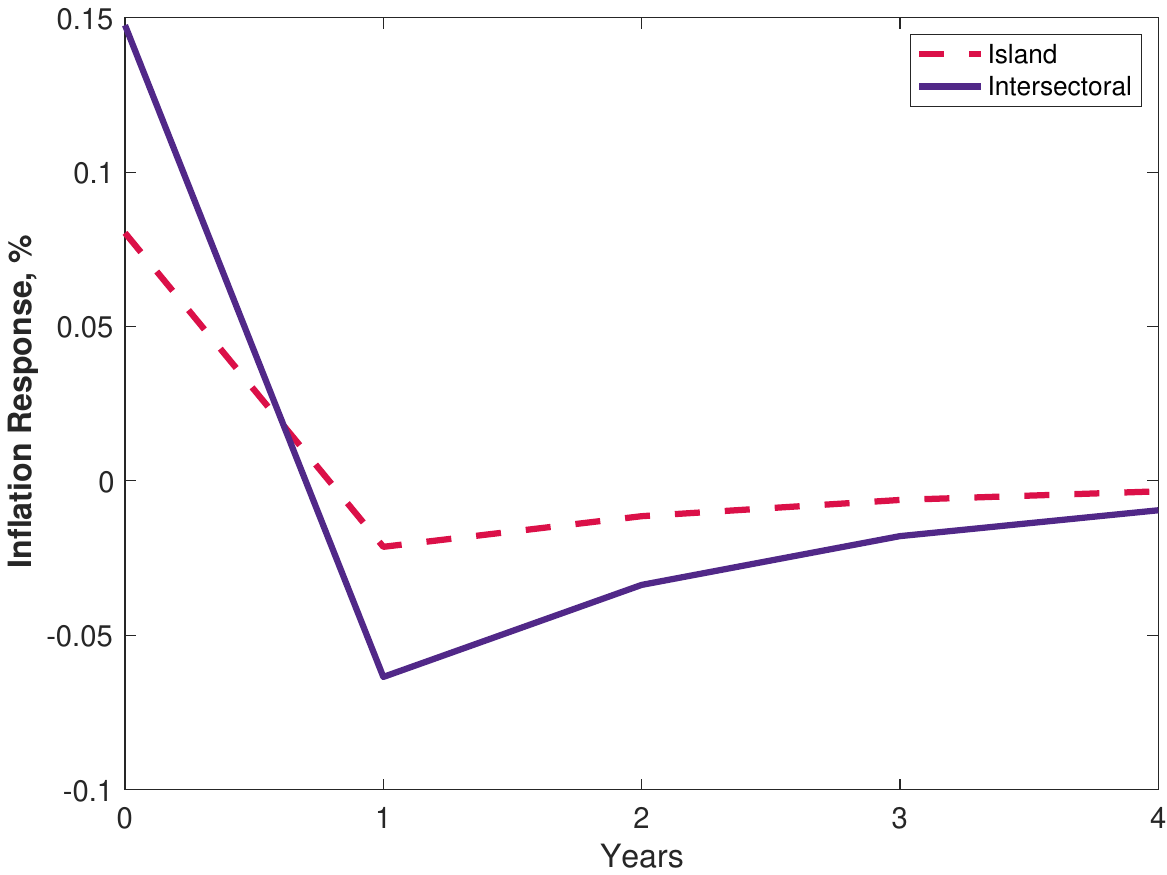}
    \end{minipage}
\Fignote{This figure shows the response of inflation to a one percent negative productivity shock (panel a) and a one percent positive import price shock (panel b). The solid purple line is for the model that considers intersectoral linkages, while the pink dashed line is the model that assumes an island production network structure.}
\end{figure}

\newpage
\section{Shares}\label{app:factor_share_solving}
In this appendix I explicitly solve for $\widehat{\bar{\bm{\Lambda}}}$. Before doing so, I need to define several objects on the consumption and production sides.
\subsection{Consumption}

Let the consumer's share expenditure on good $i \in N\cup M$ be $b_i$ and 
\begin{align*}
    \bar{b}_i&= \frac{P_iC_i}{E}
\end{align*}
Let the price elasticity of demand be $\varepsilon^C_{ik}= \frac{\partial \log C_i}{\partial \log P_k}$ and $\delta_{ik} = 1$ if $k = i$ and zero otherwise. This last element is usually called the Kronecker-delta. Log-differentiating the shares and using the homotheticity assumption, I have 
\begin{align*}
   \widehat{\bar{b}}_i &= \sum\limits_{k\in N \cup M} (\delta_{ik}+ \varepsilon^C_{ik} - \bar{b}_i)\widehat{P}_k = \sum\limits_{k\in N\cup M}\phi^C_{ik}\widehat{P}_k\\
   \mathrm{d}\bar{b}_i &= \bar{b}_i \sum\limits_{k\in N\cup M}\phi^C_{ik}\widehat{P}_k\qquad \text{for }i\in N\cup M
\end{align*}
where $\phi^C_{ik} = (\delta_{ik}+ \varepsilon^C_{ik} - \bar{b}_k)$ represents the elasticity of consumption share on good $i$, $b_i$, in response to a change in the price of good $k$, $P_k$.

Note that we must have 
\begin{align*}
    \sum\limits_{i\in N\cup M}\mathrm{d}\bar{b}_{i} = 0\Longrightarrow \sum\limits_{i\in N\cup M} \bar{b}_i \sum\limits_{k\in N\cup M}\phi^C_{ik}\widehat{P}_k=0,
\end{align*}
for any changes in prices. It thus follows that 
\begin{align*}
    \sum\limits_{i\in N\cup M} \bar{b}_i \phi^C_{ik}=0\quad \text{ for all } k \in N\cup M
\end{align*}
For further reference, it proves useful to define changes in domestic expenditure shares as
\begin{align*}
    \mathrm{d}\bar{\bm{b}}_D&= diag(\bm{\bar{b}}_D)(\bm{\Phi}^C_D\widehat{\bm{P}}_D + \bm{\Phi}^C_M\widehat{\bm{P}}_M),
\end{align*}
where $\bm{\Phi}^C_D$ is an $N\times N$ matrix with typical element $\phi^C_{ij}$ with $i,j\in N$, and $\bm{\Phi}^C_M$ is an $N\times M$ matrix with typical element $\phi^C_{im}$, and $i \in N$, $m \in M$.
\subsection{Production}
On the production side, I must define an operator similar to the above but for each $i \in N$ producer. In addition, the decision of the representative firms also depends on factor prices $\widehat{W}_f$. With this in mind, define the expenditure share of producer $i$ on input $j \in N\cup M$ as 
\begin{align*}
    \Omega_{ij}&= \frac{P_jM_{ij}}{P_iQ_i}
\end{align*}
Define 
\begin{align*}
    \phi^{i}_{jk} = \frac{\partial \log \Omega_{ij}}{\partial \log P_k}\quad \text{ for } i = 1, ..., N;\quad k = 1,..., N\cup M \cup F
\end{align*}
Then 
\begin{align*}
    \widehat{\Omega}_{ij}&= \widehat{P}_j + \sum\limits_{k\in N\cup M}\varepsilon^i_{jk} \widehat{P}_k + \sum\limits_{f\in F}\varepsilon^i_{jf} \widehat{W}_f - \sum\limits_{k\in N\cup M}\Omega_{ik} \widehat{P}_k - \sum\limits_{f\in F}a_{if}\widehat{W}_f\\
    &= \sum\limits_{k\in N\cup M}(\delta_{jk} + \varepsilon^i_{jk} - \Omega_{ik}) \widehat{P}_k + \sum\limits_{f\in F}(\varepsilon^i_{jf} -a_{if})  \widehat{W}_f\\
    \widehat{\Omega}_{ij} &= \sum\limits_{k\in N\cup M} \phi^i_{jk}\widehat{P}_k + \sum\limits_{f\in F}\phi^i_{jf}\widehat{W}_f\\
    \mathrm{d}\Omega_{ij}&= \Omega_{ij} \left(\sum\limits_{k\in N\cup M} \phi^i_{jk}\widehat{P}_k + \sum\limits_{f\in F}\phi^i_{jf}\widehat{W}_f\right)
\end{align*}
where 
\begin{align*}
    \varepsilon^i_{jk} &= \frac{\partial \log M_{ij}}{\partial \log P_k}\\
    \phi^i_{jk} &= \delta_{jk} + \varepsilon^i_{jk} - \Omega_{ik}\\
    \phi^i_{jf} &= \varepsilon^i_{jf} - a_{if}
\end{align*}
represents the elasticity of expenditure share on good $j$ by producer $i$, $\Omega_{ij}$, when there is a change in either good $k \in N \cup M$ or factors $f \in F$.

By a similar logic, I can write the change in expenditure share of producer $i$ on factor $f$ as 
\begin{align*}
    \mathrm{d}a_{if}&= a_{if}\left(\sum\limits_{k\in N\cup M}(\varepsilon^i_{fk} - \Omega_{ik}) \widehat{P}_k + \sum\limits_{f'\in F}(\delta_{ff'}+\varepsilon^i_{ff'} -a_{if'})  \widehat{W}_{f'}\right)\\
    \mathrm{d}a_{if}&=a_{if}\left(\sum\limits_{k\in N\cup M}\phi^i_{fk} \widehat{P}_k + \sum\limits_{f'\in F}\phi^i_{ff'} \widehat{W}_{f'}\right)
\end{align*}
where again 
\begin{align*}
    \varepsilon^i_{fk}&= \frac{\partial \log L_{if}}{\partial \log P_k}\\
    \varepsilon^i_{ff'}&= \frac{\partial \log L_{if}}{\partial \log W_{f'}}\\
    \phi^i_{fk} &= \varepsilon^i_{fk} - \Omega_{ik}\\
    \phi^i_{ff'} &= \delta_{ff'} + \varepsilon^i_{ff'} - a_{if'}
\end{align*}
The first two rows represent the demand elasticity of factor $f$ relative to a change in other good prices (first row) or factors of producton (second row).

The last two rows represent the elasticity of expenditure share of producer $i$ on factor $f$ relative to either good or factor price changes. When these are positive, then expenditure changes increase after a change in other input prices, meaning that the producer substitutes away from those price increases towards factor $f$. If this term is negative, then the expenditure share in factor $f$ declines with a change in other input prices: it means that it has to move resources away from factor $f$ towards those goods that are seeing an increase in their price. This is the complementarity in the production case, and it arises with low elasticities of substitution (low $\varepsilon$'s).

For imported intermediate, I can construct the same as 
\begin{align*}
    \mathrm{d}\Gamma_{im}&= \Gamma_{im}\left(\sum\limits_{k\in N}(\varepsilon^i_{mk} - \Omega_{ik}) \widehat{P}_k + \sum\limits_{m'\in M}(\delta_{mm'}+\varepsilon^i_{mm'} - \Gamma_{im'}) \widehat{P}_{m'} + \sum\limits_{f\in F}(\varepsilon^i_{mf} -a_{if})  \widehat{W}_{f}\right)\\
    \mathrm{d}\Gamma_{im}&= \Gamma_{im}\left(\sum\limits_{k\in N}\phi_{mk}^i \widehat{P}_k + \sum\limits_{m'\in M}\phi_{mm'}^i \widehat{P}_{m'} + \sum\limits_{f\in F}\phi_{mf}^i  \widehat{W}_{f}\right)
\end{align*}
where again 
\begin{align*}
    \varepsilon^i_{mk}&= \frac{\partial \log M_{im}}{\partial \log P_k}\\
    \varepsilon^i_{mm'}&= \frac{\partial \log M_{im}}{\partial \log P_{m'}}\\
    \varepsilon^i_{mf}&= \frac{\partial \log M_{im}}{\partial \log W_{f}}\\
    \phi_{mk}^i&= (\varepsilon^i_{mk} - \Omega_{ik}) \\
    \phi_{mm'}^i&= (\delta_{mm'}+\varepsilon^i_{mm'} - \Gamma_{im'})\\
    \phi_{mf}^i&= (\varepsilon^i_{mf} -a_{if})
\end{align*}

\subsection{Market clearing conditions and substitution patterns}
\paragraph{Goods market clearing conditions.}
Recall that the market clearing conditions for the $N$ domestic goods can be written as 
\begin{align*}
    Q_i &= C_i + X_i + \sum\limits_{j\in N}M_{ji}
\end{align*}
In nominal terms and dividing by expenditure, I have 
\begin{align*}
    \frac{P_iQ_i}{E} &=  \frac{P_iC_i}{E} +  \frac{P_iX_i}{E} +  \sum\limits_{j\in N}\frac{P_iM_{ji}}{P_jQ_j}\frac{P_jQ_j}{E}
\end{align*}

Define \emph{expenditure-based} ratios with a bar i.e. $\bar{\lambda}_i = \frac{P_iQ_i}{E}$. Then, 
\begin{align*}
    \bar{\lambda}_i &= \bar{b}_i + \bar{x}_i +  \sum\limits_{j\in N} \Omega_{ji}\bar{\lambda}_j
\end{align*}
Differentiating this expression, I have 
\begin{align*}
    \mathrm{d}\bar{\lambda}_i &= \mathrm{d}\bar{b}_i + \mathrm{d}\bar{x}_i +  \sum\limits_{j\in N} (\mathrm{d}\Omega_{ji}\bar{\lambda}_j + \Omega_{ji}\mathrm{d}\bar{\lambda}_j)
\end{align*}
Now, recall from shares and making some changes of indices
\begin{align*}
    \mathrm{d}\Omega_{ji}&= \Omega_{ji} \left(\sum\limits_{k\in N\cup M} \phi^j_{ik}\widehat{P}_k + \sum\limits_{f\in F}\phi^j_{if}\widehat{W}_f\right)   
\end{align*}
Then, the third term on the right-hand side can be written as
\begin{align*}
    \sum\limits_{j\in N} \mathrm{d}\Omega_{ji}\bar{\lambda}_j   &= \sum\limits_{j\in N} \Omega_{ji}\bar{\lambda}_j \left(\sum\limits_{k\in N\cup M} \phi^j_{ik}\widehat{P}_k + \sum\limits_{f\in F}\phi^j_{if}\widehat{W}_f\right)   \\
    &= \sum\limits_{j\in N} \Omega_{ji}\bar{\lambda}_j \sum\limits_{k\in N\cup M} \phi^j_{ik}\widehat{P}_k + \sum\limits_{j\in N} \Omega_{ji}\bar{\lambda}_j\sum\limits_{f\in F}\phi^j_{if}\widehat{W}_f \\
    &= \sum\limits_{j\in N} \Omega_{ji}\bar{\lambda}_j \sum\limits_{k\in N} \phi^j_{ik}\widehat{P}_k+ \sum\limits_{j\in N} \Omega_{ji}\bar{\lambda}_j \sum\limits_{m\in M} \phi^j_{im}\widehat{P}_m+ \sum\limits_{j\in N} \Omega_{ji}\bar{\lambda}_j\sum\limits_{f\in F}\phi^j_{if}\widehat{W}_f \\
    &= \sum\limits_{k\in N} \underbrace{\left[ \sum\limits_{j\in N} \Omega_{ji}\bar{\lambda}_j \phi^j_{ik}\right]}_{\equiv \phi_{ik}}\widehat{P}_k+ \sum\limits_{m\in M}\underbrace{\left[\sum\limits_{j\in N} \Omega_{ji}\bar{\lambda}_j  \phi^j_{im}\right]}_{\equiv \phi_{im}}\widehat{P}_m+ \sum\limits_{f\in F}\underbrace{\left[\sum\limits_{j\in N} \Omega_{ji}\bar{\lambda}_j\phi^j_{if}\right]}_{\equiv \phi_{if}}\widehat{W}_f\\
    &= \sum\limits_{k\in N} \phi_{ik}\widehat{P}_k+ \sum\limits_{m\in M}\phi_{im}\widehat{P}_m+ \sum\limits_{f\in F}\phi_{if}\widehat{W}_f
\end{align*}
A useful thing about writing this in this way, is that I can write this in matrix form 
\begin{align*}
    \mathrm{d}\bm{\Omega}^T \bm{\lambda} &= \bm{\Phi}_D \widehat{\bm{P}}_D + \bm{\Phi}_M \widehat{\bm{P}}_M +\bm{\Phi}_F \widehat{\bm{W}}
\end{align*}
where $\bm{\Phi}$ represents \emph{direct substitution matrices}. A version of these substitution matrices appears in \cite{BF19}, although they consider both direct and indirect substitution. Each column represents the changing price and the rows represent where intermediate input demand is going. This only takes into account \emph{first-round effects} and does not consider any input-output linkages beyond the direct exposure (or path of order 1). The next step is to recompute these matrices using the Leontief-inverse. To see this, note that the differentiated form of the market clearing condition write the problem as 
\begin{small}
\begin{align*}
    \mathrm{d}\bm{\bar{\lambda}} &= \bm{\Psi}^T(\mathrm{d}\bm{\bar{b}}_D + \mathrm{d}\bm{\bar{x}} + \bm{\Phi}_D \widehat{\bm{P}}_D + \bm{\Phi}_M \widehat{\bm{P}}_M +\bm{\Phi}_F \widehat{\bm{W}})\\
    &= \bm{\Psi}^T\left(diag(\bm{\bar{b}}_D)(\bm{\Phi}^C_D\widehat{\bm{P}}_D + \bm{\Phi}^C_M\widehat{\bm{P}}_M) + \mathrm{d}\bm{\bar{x}} + \bm{\Phi}_D\widehat{\bm{P}}_D + \bm{\Phi}_M \widehat{\bm{P}}_M +\bm{\Phi}_F \widehat{\bm{W}} \right)\\
    \mathrm{d}\bm{\bar{\lambda}}  &= \bm{\Psi}^T\left(diag(\bm{\bar{b}}_D) \bm{\Phi}^{C}_D+ \bm{\Phi}_D \right) \widehat{\bm{P}}_D + \left(diag(\bm{\bar{b}}_D) \bm{\Phi}^{C}_M + \bm{\Phi}_M\right) \widehat{\bm{P}}_M + \bm{\Psi}^T\mathrm{d}\bm{\bar{x}} +\bm{\Psi}^T\bm{\Phi}_F \widehat{\bm{W}}
\end{align*} 
\end{small}
\paragraph{Factor shares changes.}
We need $F$ more equations, which come from the factor market clearing conditions 
\begin{align*}
    \bar{L}_f &= \sum\limits_{i \in N}L_{if}
\end{align*}
Write this in share form 
\begin{align*}
    \bar{\Lambda}_f = \frac{W_f\bar{L}_f}{E} &= \sum\limits_{i \in N}\frac{W_fL_{if}}{P_iQ_i}\frac{P_iQ_i}{E} = \sum\limits_{i \in N}a_{if}\bar{\lambda}_i
\end{align*}
In differential form, 
\begin{align*}
   \mathrm{d}\Lambda_f &= \sum\limits_{i \in N}\mathrm{d}a_{if}\bar{\lambda}_i  + \sum\limits_{i \in N}a_{if}\mathrm{d}\bar{\lambda}_i
\end{align*}
Using the expression for changes in factor usage at the producer level, I have 
\begin{align*}
    \mathrm{d}\Lambda_f &= \sum\limits_{i \in N}\left[a_{if}\left(\sum\limits_{k\in N\cup M}\phi^i_{fk} \widehat{P}_k + \sum\limits_{f'\in F}\phi^i_{ff'} \widehat{W}_{f'}\right)\right]\bar{\lambda}_i  + \sum\limits_{i \in N}a_{if}\mathrm{d}\bar{\lambda}_i\\
    \mathrm{d}\Lambda_f &= \sum\limits_{i \in N}\left[a_{if}\left(\sum\limits_{k\in N}\phi^i_{fk} \widehat{P}_k +\sum\limits_{m\in M}\phi^i_{fm} \widehat{P}_m + \sum\limits_{f'\in F}\phi^i_{ff'} \widehat{W}_{f'}\right)\right]\bar{\lambda}_i  + \sum\limits_{i \in N}a_{if}\mathrm{d}\bar{\lambda}_i
\end{align*}
Taking each of the terms on the right-hand side, I can write
\begin{align*}
    \sum\limits_{i \in N}a_{if}\sum\limits_{k\in N}\phi^i_{fk} \widehat{P}_k \bar{\lambda}_i &= \sum\limits_{k\in N}\underbrace{\left(\sum\limits_{i \in N}a_{if}\phi^i_{fk}  \bar{\lambda}_i\right)}_{\equiv \phi_{fk}}\widehat{P}_k = \sum\limits_{k\in N}\phi_{fk}\widehat{P}_k\\
    \sum\limits_{i \in N}a_{if}\sum\limits_{f'\in F}\phi^i_{ff'} \widehat{W}_{f'} \bar{\lambda}_i &=  \sum\limits_{f'\in F}\underbrace{\left(\sum\limits_{i \in N}a_{if}\bar{\lambda}_i\phi^i_{ff'}\right)}_{\equiv \phi_{ff'}} \widehat{W}_{f'} = \sum\limits_{f'\in F}\phi_{ff'}\widehat{W}_{f'} \\
    \sum\limits_{i \in N}a_{if}\sum\limits_{m\in M}\phi^i_{fm} \widehat{P}_m \bar{\lambda}_i&= \sum\limits_{m\in M}\underbrace{\left(\sum\limits_{i \in N}a_{if}\bar{\lambda}_i\phi^i_{fm}\right)}_{\equiv \phi_{fm}} \widehat{P}_m = \sum\limits_{m\in M}\phi_{fm}\widehat{P}_m
\end{align*}
Replacing this into the differential form for the factor share 
\begin{align*}
    \mathrm{d}\Lambda_f  &=   \sum\limits_{k\in N}\phi_{fk}\widehat{P}_k +  \sum\limits_{f'\in F}\phi_{ff'}\widehat{W}_{f'} + \sum\limits_{m\in M}\phi_{fm}\widehat{P}_m + \sum\limits_{i \in N}a_{if}\mathrm{d}\bar{\lambda}_i
\end{align*}
In matrix form,
\begin{align*}
    \mathrm{d}\bm{\bar{\Lambda}}&= \bm{\Phi}^F_{D}\widehat{\bm{P}}_D + \bm{\Phi}^{F}_F\widehat{\bm{W}} + \bm{\Phi}^{F}_M\widehat{\bm{P}}_M + \bm{A}^T \mathrm{d}\bm{\bar{\lambda}}
\end{align*}

\paragraph{Export share changes.} Since export demand is exogenous, I can write 
\begin{align*}
    \bar{x}_i &= \frac{P_iX_i}{\mathcal{M}}\Longrightarrow \mathrm{d}\bar{x}_i = \bar{x}_i(\widehat{P}_i + \widehat{X}_i - \widehat{\mathcal{M}}),
\end{align*}
which staking into a vector form 
\begin{align}
    \mathrm{d}\bm{\bar{x}} &= diag(\bm{\bar{x}})(\widehat{\bm{P}}_D + \widehat{\bm{X}} - \bm{1}_N\widehat{\mathcal{M}})
\end{align}
\subsection{Solving the system}
\begin{footnotesize}
\begin{align*}
    \mathrm{d}\bm{\bar{\Lambda}}&= \bm{\Phi}^F_{D}\widehat{\bm{P}}_D + \bm{\Phi}^{F}_F\widehat{\bm{W}} + \bm{\Phi}^{F}_M\widehat{\bm{P}}_M + \bm{A}^T \mathrm{d}\bm{\bar{\lambda}} \qquad (F\text{ equations}, F + N + F + N \text{ unknowns})\\
     \mathrm{d}\bm{\bar{\lambda}}  &= \bm{\Psi}^T\left(diag(\bm{\bar{b}}_D) \bm{\Phi}^{C}_D+ \bm{\Phi}_D \right) \widehat{\bm{P}}_D + \left(diag(\bm{\bar{b}}_D) \bm{\Phi}^{C}_M + \bm{\Phi}_M\right) \widehat{\bm{P}}_M + \bm{\Psi}^T\mathrm{d}\bm{\bar{x}} +\bm{\Psi}^T\bm{\Phi}_F \widehat{\bm{W}}\qquad (N \text{ equations}, N \text{ add. unknowns})\\
    \widehat{\bm{P}}_D &= -\bm{\Psi}\widehat{\bm{Z}} + \bm{\Psi}\bm{A}\widehat{\bm{W}} + \bm{\Psi}\bm{\Gamma}\widehat{\bm{P}}_M \qquad (N \text{ equations, no new unknowns})\\
    \mathrm{d}\bm{\bar{\Lambda}}&= diag(\bm{\bar{\Lambda}})\left(\widehat{\bm{W}} + \widehat{\bm{L}} - \bm{1}_F\widehat{\mathcal{M}} \right)\qquad  (F \text{ equations, no new unknowns})\\
    \mathrm{d}\bm{\bar{x}} &= diag(\bm{\bar{x}})(\widehat{\bm{P}}_D + \widehat{\bm{X}} - \bm{1}_N\widehat{\mathcal{M}})\qquad (N \text{ equations, no new unknowns})
\end{align*}
\end{footnotesize}
This is a system of $2F + 3N$ unknowns: ($\mathrm{d}\bm{\bar{\Lambda}}, \widehat{\bm{W}}, \mathrm{d}\bm{\bar{\lambda}}, \widehat{\bm{P}}_D, \mathrm{d}\bm{\bar{x}}) $ on the same number of equations, and thus it pins down all necessary objects.

Note that the distribution of factor shares ($\mathrm{d}\bm{\bar{\Lambda}}$) and Domar weights ($\mathrm{d}\bm{\bar{\lambda}}$) changes crucially depends on the substitution matrices ($\bm{\Phi}'s$ matrices). As a result, to the extent that substitution patterns are encapsulated in these matrices, they affect aggregate inflation in the small open economy via changing factor shares. It is in this sense that elasticities of substitution matters for inflation in this model, something that does not hold in the closed economy as these terms cancels out, to a first-order.
\end{document}